\documentclass[runningheads]{llncs}
\usepackage{graphicx}
\usepackage{algorithm}
\usepackage[noend]{algpseudocode}
\algnewcommand{\lIf}[1]{\State\algorithmicif\ #1\ \algorithmicthen}

\usepackage{newfloat}
\usepackage{listings}
\usepackage{amsmath}
\usepackage{lipsum}
\usepackage{caption}
\usepackage{subcaption}
\usepackage{xcolor}
\usepackage{multirow, makecell}
\usepackage{pbox}
\usepackage{booktabs}
\usepackage{tablefootnote}

\usepackage[firstpage]{draftwatermark}

\usepackage{amssymb,amsmath,amsfonts}
\allowdisplaybreaks

\newcommand{\HC}{\mathcal{H}}

\newcommand{\CF}{\mathcal{K}}
\newcommand{\LF}{L}
\newcommand{\CII}{\mathcal{C}}

\newcommand{\lang}[1]{\mathcal{L}(#1)}

\newcommand{\floor}[1]{\left\lfloor #1 \right\rfloor}
\newcommand{\ceil}[1]{\left\lceil #1 \right\rceil}

\newcommand{\cost}{c}

\usepackage[colorlinks]{hyperref}

\begin{document}
\title{Randomized Synthesis for Diversity and Cost Constraints with Control Improvisation\thanks{The two first authors contributed equally to the paper.}}
\titlerunning{Randomized Synthesis for Diversity and Cost Constraints with CI}
%
\author{Andreas Gittis* \and Eric Vin* \and\\ Daniel J. Fremont}
\authorrunning{A. Gittis*, E. Vin*, and D. J. Fremont}
%
\institute{University of California, Santa Cruz, USA\\
\email{\{agittis,evin,dfremont\}@ucsc.edu}}
\maketitle              
%

\SetWatermarkAngle{0}
\SetWatermarkText{\raisebox{13.1cm}{%
{\hypersetup{hidelinks}\href{https://doi.org/10.5281/zenodo.6558391}{\includegraphics{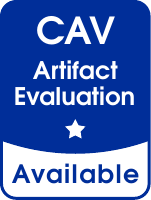}}}%
\hspace{2.75in}%
\includegraphics{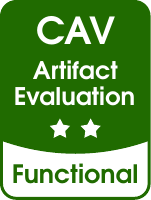}%
}}

\begin{abstract}
In many synthesis problems, it can be essential to generate implementations which not only satisfy functional constraints but are also \emph{randomized} to improve variety, robustness, or unpredictability. The recently-proposed framework of control improvisation (CI) provides techniques for the correct-by-construction synthesis of randomized systems subject to hard and soft constraints. However, prior work on CI has focused on qualitative specifications, whereas in robotic planning and other areas we often have quantitative quality metrics which can be traded against each other. For example, a designer of a patrolling security robot might want to know by how much the average patrol time needs to be increased in order to ensure that a particular aspect of the robot's route is sufficiently diverse and hence unpredictable. In this paper, we enable this type of application by generalizing the CI problem to support quantitative soft constraints which bound the expected value of a given cost function, and randomness constraints which enforce diversity of the generated traces with respect to a given label function. We establish the basic theory of labelled quantitative CI problems, and develop efficient algorithms for solving them when the specifications are encoded by finite automata. We also provide an approximate improvisation algorithm based on constraint solving for any specifications encodable as Boolean formulas. We demonstrate the utility of our problem formulation and algorithms with experiments applying them to generate diverse near-optimal plans for robotic planning problems.
\end{abstract}

\section{Introduction}
\label{sec:intro}

Correct-by-construction synthesis of systems from high-level specifications has become a popular paradigm in fields ranging from circuit design~\cite{BLOEM20073} to robotic task planning~\cite{kress2009temporal}.
Synthesis techniques for many different types of specifications have been developed, especially for temporal logic formulas, which can encode many properties of interest~\cite{finkbeiner-survey}.
One less-studied type of specification are \emph{randomness constraints} that require the system's behavior to be sufficiently random, for instance by being close to a uniform distribution over the set of allowed behaviors.
Such specifications are useful in many applications, as randomness can provide robustness, variety, and unpredictability to a system.
For example, fuzz testing tools often use constraints to select classes of inputs which are more likely to trigger bugs, but then search randomly within that class to prevent bias~\cite{fuzzing-book}.
In robotic planning, a patrolling security robot that uses a fixed plan satisfying its requirements might be vulnerable to exploitation; adding randomness to make its route unpredictable can make exploitation more difficult.

While there has been substantial work on synthesis with stochastic environments (e.g.~\cite{multiobjective,almagor-kupferman}), randomness constraints require the system itself to behave randomly even if the environment is deterministic.
Furthermore, unlike most specifications used in synthesis, randomness constraints are properties not of individual behaviors but rather of their distribution, and they cannot be concisely encoded into existing specification formalisms like PCTL~\cite{pctl} and SGL~\cite{sgl}.
As a result, synthesis of systems under such constraints requires new techniques.

A recently-proposed paradigm for the correct-by-construction synthesis of systems under randomness constraints is \emph{algorithmic improvisation}~\cite{donze:ci_music_paper,fremont:ci_original,fremont:thesis}.
Algorithmic improvisation comprises a class of synthesis problems whose goal is to construct a randomized algorithm, an \emph{improviser}, satisfying three kinds of constraints: \emph{hard constraints} that the improviser's output must always satisfy, \emph{soft constraints} that need only be satisfied to a certain (tunable) extent, and \emph{randomness constraints} requiring the output to be sufficiently random.
These types of constraints correspond to natural requirements arising for example in robot planning: the hard constraints can encode safety or other functional requirements, the soft constraints can encode notions of efficiency or optimality, and the randomness constraints enforce diversity or unpredictability.
The original and most-studied form of algorithmic improvisation is the \emph{control improvisation} (CI) problem (introduced in \cite{donze:ci_music_techreport} and formalized in \cite{fremont:ci_original,fremont:ci_extended}), where the improviser generates finite sequences of symbols, the hard constraint is a trace property, the soft constraint requires some trace property hold with at least a desired probability, and the randomness constraint puts upper and lower bounds on the probability of individual outputs.
Control improvisation and its extensions have been successfully used for musical improvisation~\cite{donze:ci_music_paper}, robotic planning~\cite{fremont:reactive_ci}, human modeling subject to constraints~\cite{ci_iotdi}, and generating synthetic datasets for testing and training cyber-physical systems with machine learning components~\cite{scenic}.

However, the prior work on CI is not general enough to cover many randomized synthesis problems of interest, for two reasons.
First, many planning, design space exploration, and other problems come with a \emph{cost function} expressing how optimal a particular solution is; in the setting of generating randomized solutions, the most natural soft constraint would be to require that the \emph{expected cost} of a solution should be low, so that we can obtain a diverse set of near-optimal solutions.
In a patrolling robot application, for example, the fastest patrol route might be unique and so predictable, and we then want to know by how much we would need to increase the average patrol time in order to enable a sufficiently-diverse set of routes.
The prior work on CI cannot provide such an analysis.

Second, while the CI randomness constraint is sufficient to make the improviser's exact output unpredictable, it is \emph{not} sufficient to ensure diversity when many outputs are similar to each other.
Continuing our patrolling robot example, suppose that the robot has a choice of two rooms to go through: one room is larger, and so there are (say) $10^6$ possible paths through it, vs. only $10^3$ through the other room.
Even if a perfectly-uniform distribution over all these paths is possible given our other constraints, the robot will end up entering the larger room almost all of the time.
But from the point of view of an adversary that wishes to avoid being seen by the robot, the exact path is not relevant: what matters is \emph{which room} the robot will enter, and that is highly predictable.
For this application, we need a randomness requirement that enforces diversity not over the output of the improviser, but over some \emph{attribute} of the output.

To enable such applications, in this paper we introduce the concept of \emph{Labelled Quantitative Control Improvisation} (LQCI).
This problem extends CI with a soft constraint bounding the expected cost of generated traces, and a randomness constraint requiring near-uniformity of the \emph{label} of a trace, given by an arbitrary label function.
We study the theory of LQCI, establishing precise conditions for when an LQCI problem is solvable and a general construction for solving it.
We use our construction to develop efficient improvisation algorithms for a broad class of specifications given by finite automata, including common cost functions such as mission time or path length.
For specifications not easily encoded to (reasonably-sized) automata, we provide an approximate improvisation algorithm based on constraint solving that handles symbolic specifications encoded as Boolean formulas.
We also explore an extension of the LQCI problem for finding the \emph{maximum-entropy} distribution satisfying the other constraints (as in \cite{vazquezchanlatte:maximum_entropy_ci}), and develop an algorithm for solving it using convex optimization.
Finally, we conduct a case study demonstrating that our approach allows us to formalize and solve realistic robotic planning problems.

In summary, the main contributions of this paper are:
\begin{itemize}
    \item The labelled quantitative control improvisation problem definition (Sec.~\ref{sec:prob_def});
    \item A characterization of which LQCI problems are solvable, and a general construction for solving them (Sec.~\ref{sec:theory});
    \item Efficient improvisation algorithms for finite automata specifications (Sec.~\ref{sec:exact-schemes});
    \item An approximate algorithm for Boolean formula specifications (Sec.~\ref{sec:approx_lqci_alg});
    \item An algorithm for maximum-entropy LQCI problems (Sec.~\ref{sec:maxent});
    \item Experiments using our algorithms for robotic planning (Sec.~\ref{sec:experiments}).
\end{itemize}
We conclude in Sec.~\ref{sec:conclusion} with a summary of results and directions for future work.
For brevity, we defer full proofs of all results to the Appendix.

\section{Overview and Problem Definition}
\label{sec:prob_def}

In this section we formally define the LQCI problem, first using applications to robotic planning and fuzz testing to motivate various aspects of our definitions.
We will return to the robotic planning example for our experiments in Sec.~\ref{sec:experiments}.

\subsection{Motivating Examples}

\subsubsection{Robotic Planning}

Consider the problem of generating a path for a package delivery robot, where the robot should efficiently visit various drop-off points, visiting charging stations as necessary along the way.
Discretizing the world into a grid, we can represent a path as a finite sequence of north, south, east, and west moves.
We might have various requirements for such paths, falling into the three types of constraints of a control improvisation problem described above: hard constraints such as completing mission objectives and not navigating into impassable terrain, soft constraints such as preferring shorter paths, and randomness constraints to ensure the chosen path is unpredictable.
However, as we saw in Sec.~\ref{sec:intro}, randomness over paths can be less important than randomness over specific features of a path: here, it might be that charging leaves the robot vulnerable for an extended period, so that it is important to limit the extent to which an adversary can predict ahead of time which charging station will be used.
If there are 3 charging stations, then all possible paths are divided into 3 classes, and we want the class of a generated path to be unpredictable; we can formalize this as a \emph{label function} which assigns labels to paths, and require that the distribution over labels be close to uniform.
Since we do not want to simply pick a single path from each label class, we can also enforce randomness within each class, either by bounding the conditional probabilities of paths (so that no path is too likely relative to others in its class) or by taking the maximum-entropy distribution that satisfies our randomness-over-labels condition (we will return to this approach in Sec.~\ref{sec:maxent}).

For efficiency, we want our robot to use routes which are as fast as possible, taking into account varying terrain.
We could model this using a \emph{cost function} assigning numerical costs to each path: here, the total time needed to traverse it.
However, as mentioned in Sec.~\ref{sec:intro}, prior work on CI can only encode Boolean soft constraints, such as requiring the cost of a path to be at most $5$ with probability at least $0.9$.
While this does allow for some control over the cost, it requires setting an arbitrary threshold, and otherwise ignores the actual values of the cost; thus, a path of cost 6 is treated no differently than a path of cost $10^5$.
Instead, we want to bound the \emph{expected} cost of a path, so that both the probabilities of individual paths and their absolute costs are taken into account.


Putting all this together, we define our example planning problem as generating paths through the grid worlds in Fig.~\ref{fig:rp_grid_worlds}, subject to the following constraints:
    
    \begin{enumerate}
        \item[~] \textbf{Hard Constraint:}
            \begin{enumerate}
                \item The robot must begin in the start cell \textbf{S} and must end in the end cell \textbf{E}.
                \item The robot must visit all package drop-off points \textbf{O}.
                \item The robot must charge at a charging station \textbf{C}.
                \item The robot must not enter impassable locations \textbf{X}.
            \end{enumerate}
        \item[~] \textbf{Cost Constraint:}\\
        The expected time to complete the mission must be at most a constant $c$.
        \item[~] \textbf{Randomness over Labels:}\\
        For each choice of charging station, the chance that the robot uses that station must be at least $\lambda$ and at most $\rho$.
        \item[~] \textbf{Randomness over Words:}\\
        Conditioned on selecting a certain charging station, the probability of picking any path must be at least $\alpha$ and at most $\beta$.
    \end{enumerate}

Here, we assume that each grid cell has a cost representing how long it takes to traverse, with the cost of a path (the total mission time) being the sum of the costs of its cells.
In Fig.~\ref{fig:rp_grid_worlds}, we show higher-cost cells as being darker, with the costs ranging from 0--3 for the small world and 0--10 for the large world.
The layout of the map was chosen to admit a variety of different paths, motivated as follows: we envision an impassible river dividing the top and bottom halves of the map, with one low-cost bridge and two high-cost fords.
The top-left charging station is a windmill and requires climbing a hill to access; there is also a hydroelectric station next to the river, and an easily-accessible substation near the main north-south road.

    \begin{figure}[tb]
        \centering
        \begin{subfigure}[b]{0.33\textwidth}
            \includegraphics[width=\textwidth]{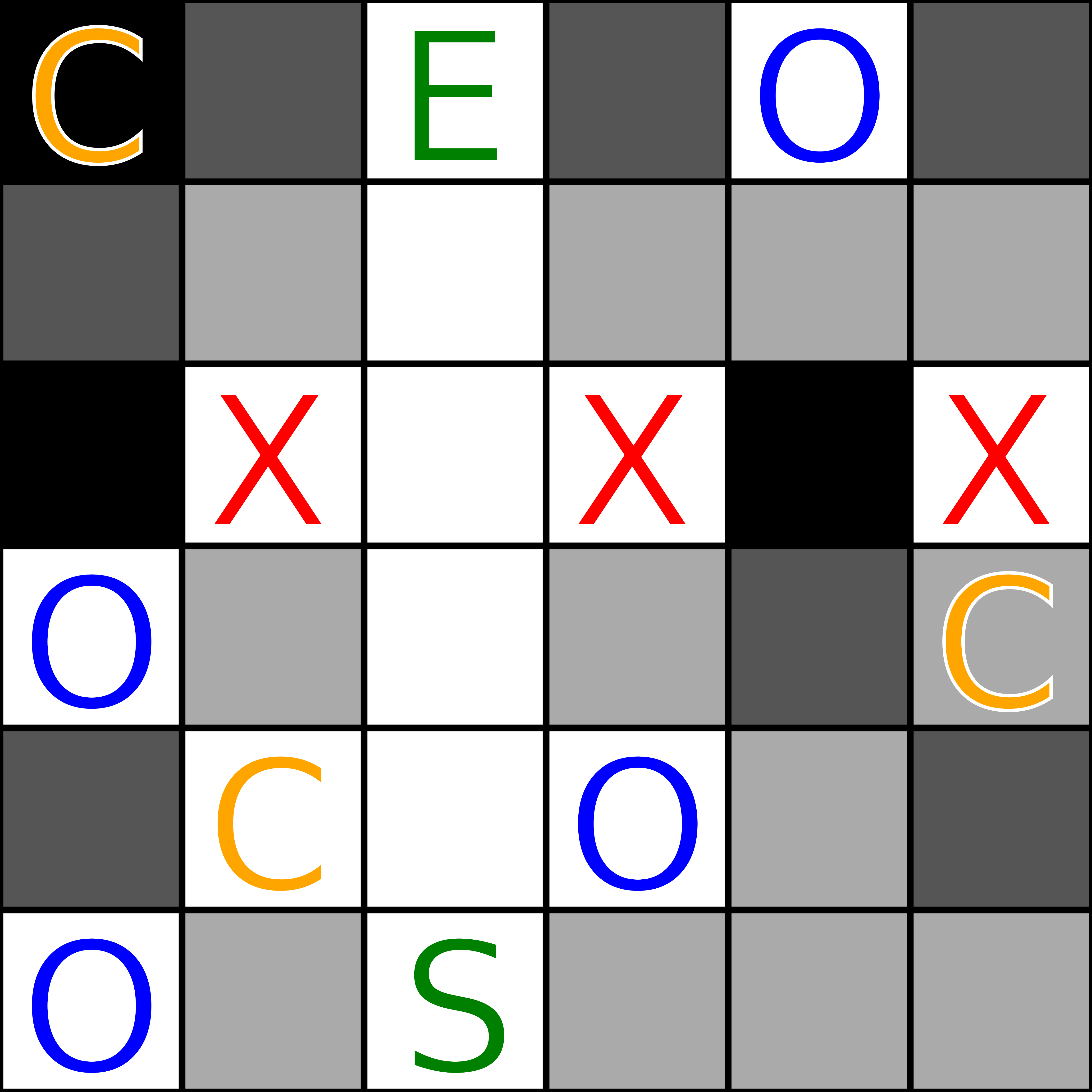}
            \caption{Small Grid World (6x6)}
        \end{subfigure}
        \hspace{1in}
        \begin{subfigure}[b]{0.33\textwidth}
            \includegraphics[width=\textwidth]{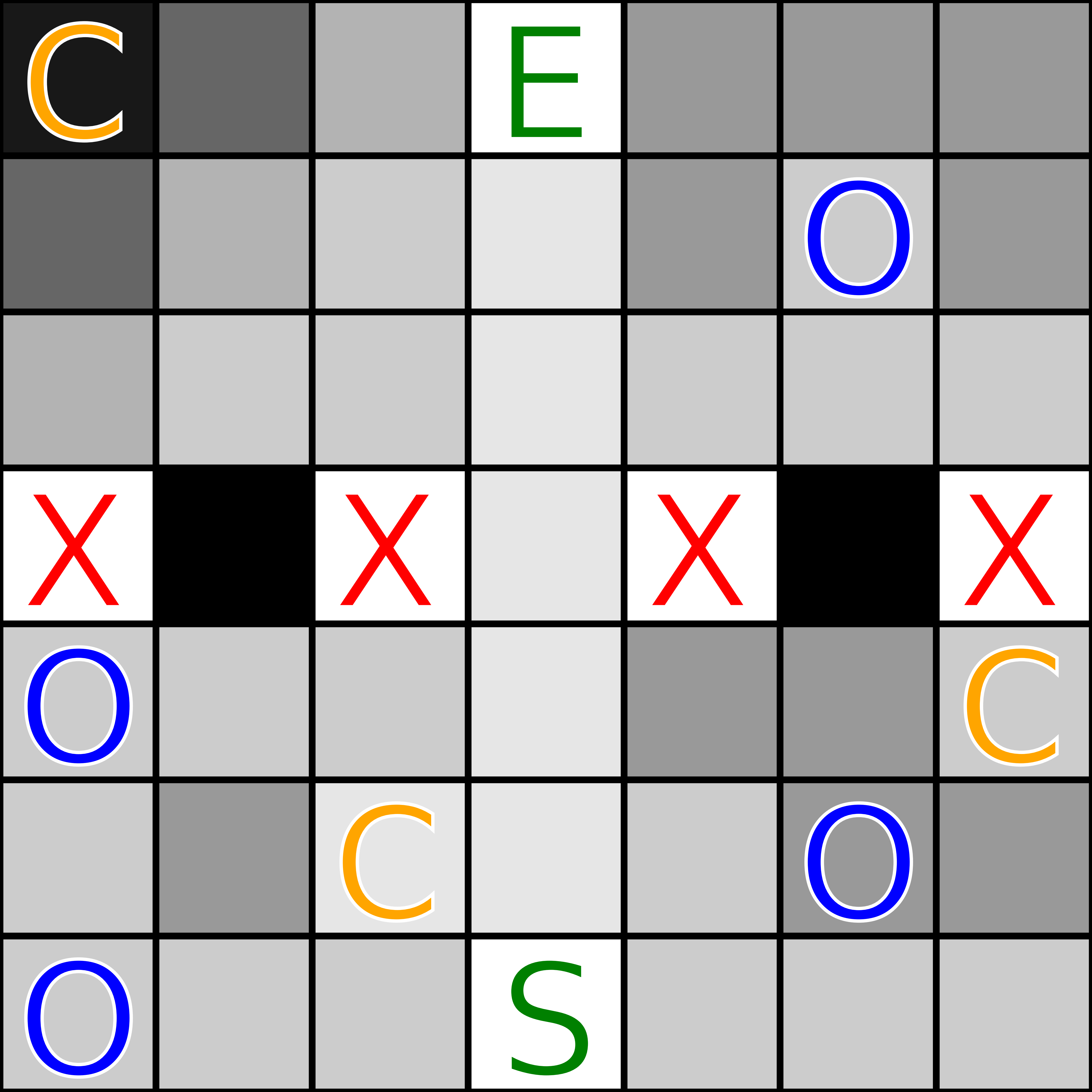}
            \caption{Large Grid World (7x7)}
        \end{subfigure}
        \caption{Grid worlds for our robotic planning example. Darker background indicates higher cost and letters indicate: start and end points (\textbf{S}, \textbf{E}), impassable locations (\textbf{X}), delivery locations (\textbf{O}), charging stations (\textbf{C}).}
        \label{fig:rp_grid_worlds}
    \end{figure}


\subsubsection{Fuzz Testing}

Prior work has shown that a variety of programs and protocols can be comprehensively tested by randomly sampling from automata encoding constraints on acceptable tests~\cite{automata_testing}.
LQCI allows us to preserve such guarantees while exercising additional control over which tests are generated.

As an example, consider the problem of generating randomized network activity for a set of devices communicating over TCP; this could be useful to test robustness of a network monitoring application or network stack.
There are a variety of different constraints we might wish to impose on the sequences of packets we generate: each connection should conform to the TCP protocol, so that the tests are meaningful\footnote{We might also want to generate tests that \emph{deviate} from the protocol. This could be done in a variety of ways, e.g. modifying our constraints to allow certain types of deviations, or first generating tests that conform to the protocol and subsequently mutating them.}; tests should exhibit a variety of different behaviors such as successful/failed connections, interleaving of packets between different connections, etc.; and tests should be as short as possible while still exhibiting these different behaviors, so that we can maximize the number of tests we can perform in a given time.
These constraints have trade-offs: for example, tests with failed connections that must be retried will necessarily be longer.
As in the robotic planning example, we formulate these requirements as cost and label constraints, which allow us to balance our randomness and control needs.

For concreteness, consider the specific example of generating packet traces for 5 systems communicating over TCP.
Our hard constraint can enforce that each connection follows the TCP protocol, using an encoding of the operation of the protocol as a finite automaton~\cite{tcp_state_machine} (we will present efficient algorithms for LQCI with automata specifications below).
Our cost function can assign a cost equal to the length of the trace, so that we prefer shorter sequences (whereas if we simply sampled uniformly from the language of the TCP automaton up to some length, longer sequences would be generated more frequently as there are exponentially more of them).
Our label function could use two labels, distinguishing traces with connections that terminate cleanly from those that involve system failures and timeouts (we could also further subdivide into several types of failures).
There are many more ways for a connection to fail than to terminate cleanly, and these two classes of traces might have significantly different lengths on average, but we want to ensure that our tests cover both cases adequately.
By imposing constraints on the expected cost of a trace, as well as randomness constraints over the label and within each label class, we can control test length while enforcing sufficient diversity among the tests.
In fact, we will see below that our LQCI algorithms can find the \emph{minimum-cost} distribution consistent with the randomness constraints, thereby allowing us to test as efficiently as possible given coverage requirements.

\subsection{Problem Definition}

To formalize synthesis problems like those described above, we define the LQCI problem.
Following the definition of CI~\cite{fremont:ci_original,fremont:ci_extended}, we frame the problem as sampling words over a finite alphabet $\Sigma$ subject to several constraints.
We use the general term \emph{specification} to refer to an encoding of a property of words (a language): for example, a deterministic finite automaton (DFA) is a specification, where the DFA accepts a word if and only if it satisfies the specification; a Boolean formula is another kind of specification.
The complexity of the LQCI problem will vary depending on the type of specifications used, as we will see later.

	\begin{definition}
    	A \emph{Labelled Quantitative Control Improvisation (LQCI) instance} over an alphabet $\Sigma$ is a tuple $\CII= (\HC, \CF, \LF, m, n, \cost, \lambda, \rho, \hat{\alpha}, \hat{\beta})$ which contains:

    	\begin{itemize}
    		\item $m,n \in \mathbb{N}$, lower and upper bounds on word length (with $m \leq n$);
	        \item $\HC$, a hard specification that must be satisfied by all words;
	        \item $\CF :\Sigma^* \rightarrow \mathbb{Q}$, a cost function mapping words to rational costs;
	        \item $\LF: \Sigma^* \rightarrow \Omega$, a label function mapping words to a finite set of labels $\Omega = \{ \ell_1, \hdots \ell_{|\Omega|} \}$;
	        \item $\cost \in \mathbb{Q}^+$, an upper bound on expected cost;
	        \item $\lambda, \rho \in \mathbb{Q}$, lower and upper bounds on the marginal probability of selecting a word with a certain label (with $0 \leq \lambda \leq \rho \leq 1$);
	        \item $\hat{\alpha}_i, \hat{\beta}_i \in \mathbb{Q}$, lower and upper bounds on the conditional probability of words in label class $\ell_i$ (with $0 \leq \hat{\alpha}_i \leq \hat{\beta}_i \leq 1$ for all $i$).
	    \end{itemize}
	\end{definition}

We note that the specifications and functions above are abstract, and our definition does not make any assumptions about how they will be encoded in a particular problem. For example, the hard constraint $\HC$ over words might be instantiated as the language of a DFA, context-free grammar, etc. Later in the paper we will develop algorithms for solving classes of LQCI instances with specification formalisms that satisfy certain properties.

The restriction to finite traces (via the length bounds $m$ and $n$) is consistent with prior work on using CI for robotic planning~\cite{fremont:reactive_ci}: we frequently want plans that complete within a time limit.
Likewise in fuzz testing we want tests of bounded length.
Furthermore, as we will see, finite-trace LQCI is still a highly nontrivial problem, so we leave its extension to infinite traces as future work.

Given an LQCI instance, we define several convenient notations:
	
	\begin{itemize}
	    \item $\Sigma^{m:n}$ is all words satisfying the length bounds: $\{ w \in \Sigma^* \mid m \leq |w| \leq n \}$.
		\item The set of \emph{improvisations} $I$ consists of all words satisfying the length bounds and the hard specification. These are all the words which our improviser is allowed to generate.
		\item Since the length bounds $m,n$ ensure $I$ is finite, we can consider the image of $I$ under $\CF$, which must also be finite. We will refer to this set of \emph{possible costs} as $\Theta = \{ \theta_1, \hdots, \theta_{|\Theta|}\}$ (note that enumerating $\Theta$ may require an algorithm).
		\item The \emph{cost class} $I_{i,k}$ consists of all words with label $\ell_i$ and cost $\theta_k$ which satisfy the length bounds and the hard specification, i.e., $\{w \in \Sigma^{m:n} \mid w \in \lang{H}, L(w) = \ell_i, \CF(w) = \theta_k\}$.
		As the costs of all words in a cost class are equal, we may speak of the \emph{cost} of a cost class without ambiguity.
		\item The \emph{label class} $I_i$ consists of all words with label $\ell_i$ as above but any cost, i.e., $\bigcup_{k = 1}^{|\Theta|} I_{i,k}$.
		\item We write $\Pr[ X(w) \mid w \leftarrow D]$ for the probability (or $E[\dots]$ for the expected value) of $X(w)$ given that $w$ is sampled from distribution $D$.

	\end{itemize}

	\begin{definition}
	\label{def:improvising_dist}
		Given an LQCI instance $\CII$, a distribution $D$ over $\Sigma^*$ is an \emph{improvising distribution} for that instance if it satisfies the following constraints:
		
		\begin{enumerate}
			\item \textbf{Hard Constraint:} $ \displaystyle \Pr[w \in I \mid w \leftarrow D] = 1$
			
			\item \textbf{Cost Constraint:} $ \displaystyle E[\CF(w) \mid w \leftarrow D] \leq c$
		
			\item \textbf{Randomness over Labels:} $ \displaystyle \forall i \in \{1, \hdots, |\Omega| \}, \ \lambda \leq \Pr{[w \in I_i \mid w \leftarrow D]} \leq \rho $
			
			\item \textbf{Randomness over Words:} $ \displaystyle \forall i \in \{1, \hdots, |\Omega| \}, \ \forall y \in I_i, \newline \ \hat{\alpha}_i \leq \Pr[y = w \mid w \in I_i, w \leftarrow D] \leq \hat{\beta}_i$
		\end{enumerate}		
		
		We say that an LQCI instance is \emph{feasible} if there exists an improvising distribution for it (and \emph{infeasible} otherwise).
		An \emph{improviser} for an LQCI instance is a probabilistic algorithm which takes no input, has finite expected runtime, and whose output distribution is an improvising distribution.
		Given an LQCI instance $\CII$, the \emph{LQCI problem} is then to determine if $\CII$ is feasible, and, if so, to generate an improviser for $\CII$.
		Finally, an \emph{improvisation scheme} for a class of LQCI instances is a probabilistic algorithm with finite expected runtime that solves the LQCI problem for instances in that class.
		\end{definition}

As described in the preceding sections, the goal of our problem definition is to provide formal guarantees about the randomness of improvisations while respecting the various constraints.
In some applications, we may simply wish to maximize randomness: then precise control over the randomness parameters for each label class is not needed, and in fact finding values of $\hat{\alpha}_i, \hat{\beta}_i$ which maximize randomness while remaining feasible is nontrivial.
Building on our analysis of the basic LQCI problem in the next several sections, in Sec.~\ref{sec:maxent} we will introduce a \emph{maximum-entropy} version of LQCI which directly maximizes randomness without requiring $\hat{\alpha}_i$ and $\hat{\beta}_i$ to be explicitly specified.

\section{Feasibility Conditions and the Greedy Construction}
\label{sec:theory}

In this section, we introduce a greedy construction which will be used to provide necessary and sufficient conditions for an LQCI instance to be feasible.
This construction will also form the basis of the improvisation schemes presented later in the paper. For now, we will present the construction without assuming any particular specification formalism and ignoring algorithmic concerns: the description presented here will consider traces one by one and thus be inefficient. The next section will develop efficient implementations of these ideas.

The \emph{greedy LQCI construction} is separated into two phases. In the first phase, the \emph{greedy cost construction}, we define a distribution over each label class individually, greedily optimizing cost by giving as much weight as we can to the cheapest elements while respecting the randomness over words condition.
In the second phase, the \emph{greedy label construction}, we define a distribution over labels, greedily assigning maximum marginal probability to the label classes with the cheapest expected costs under the distributions from the first phase while respecting the randomness over labels condition.
The intuition is that we want to first make sampling within each label class as cheap as possible, and then sample from the cheapest classes as often as possible, while satisfying the randomness requirements.
We will prove below that this greedy approach in fact yields an improvising distribution whenever one exists.

\paragraph{\textbf{Toy Example.}}
We will begin with a toy example which illustrates the idea and correctness of the greedy construction. Suppose we want to sample from words of length 3 ($m=n=3$) over the binary alphabet $\Sigma = \{0,1\}$, subject to the hard constraint that each word must contain at least one 1.
We will have two label classes: words with an odd number of $1$s will be in label 1, and those with an even number in label 2.
The cost of each word will be its integer value in binary.
The label parameters will be $\lambda = 0.2$ and $\rho = 1.0$, so that each label must be sampled from with a probability at least 0.2 and at most 1.0.
The word randomness parameters will be $\hat{\alpha}_1 = \hat{\alpha}_2 = 0.1$ and $\hat{\beta}_1 = \hat{\beta}_2 = 0.5$, so that when sampling from a particular label class, each word in the class must be selected with probability at least 0.1 and at most 0.5.

Figure~\ref{fig:greedy_construction} shows the greedy construction applied to this LQCI instance.
Beginning with label 1, we need to construct a probability distribution over the words 001, 010, 100, and 111. We start by assigning 0.1 to each word, since $\hat{\alpha}_1 = 0.1$. Then we assign as much additional probability as we can (up to $\hat{\beta}_1 = 0.5$) to the cheapest words first until a total of 1 is reached, as shown in the bottom left of Fig.~\ref{fig:greedy_construction}.
The result is that there are 3 distinct probabilities within the label class: the minimum $\hat{\alpha}_1 = 0.1$, the maximum $\hat{\beta}_1 = 0.5$, and the overflow probability $0.3$ on the word 010. This process results in a distribution over label 1 with expected cost 2.2, the minimum achievable while satisfying the randomness over words constraint. A similar process yields a distribution of expected cost 4.1 on label 2. Now that we know the minimum expected cost for each label, we should sample from the cheaper label as frequently as possible. Since $\lambda = 0.2$ and $\rho = 1.0$, we sample from label 2 with probability 0.2 (the minimum allowed) and from label 1 with probability 0.8, yielding a distribution over improvisations with expected cost 2.58.
Our analysis will show that this is in fact the minimum possible expected cost over all distributions satisfying conditions (1), (3), and (4) in Def.~\ref{def:improvising_dist}.
So if the cost bound $c$ in the LQCI instance is at least this large, then we have an improvising distribution, and otherwise the instance is infeasible.

\begin{figure}[tb]
    \centering
    \includegraphics[width=\textwidth]{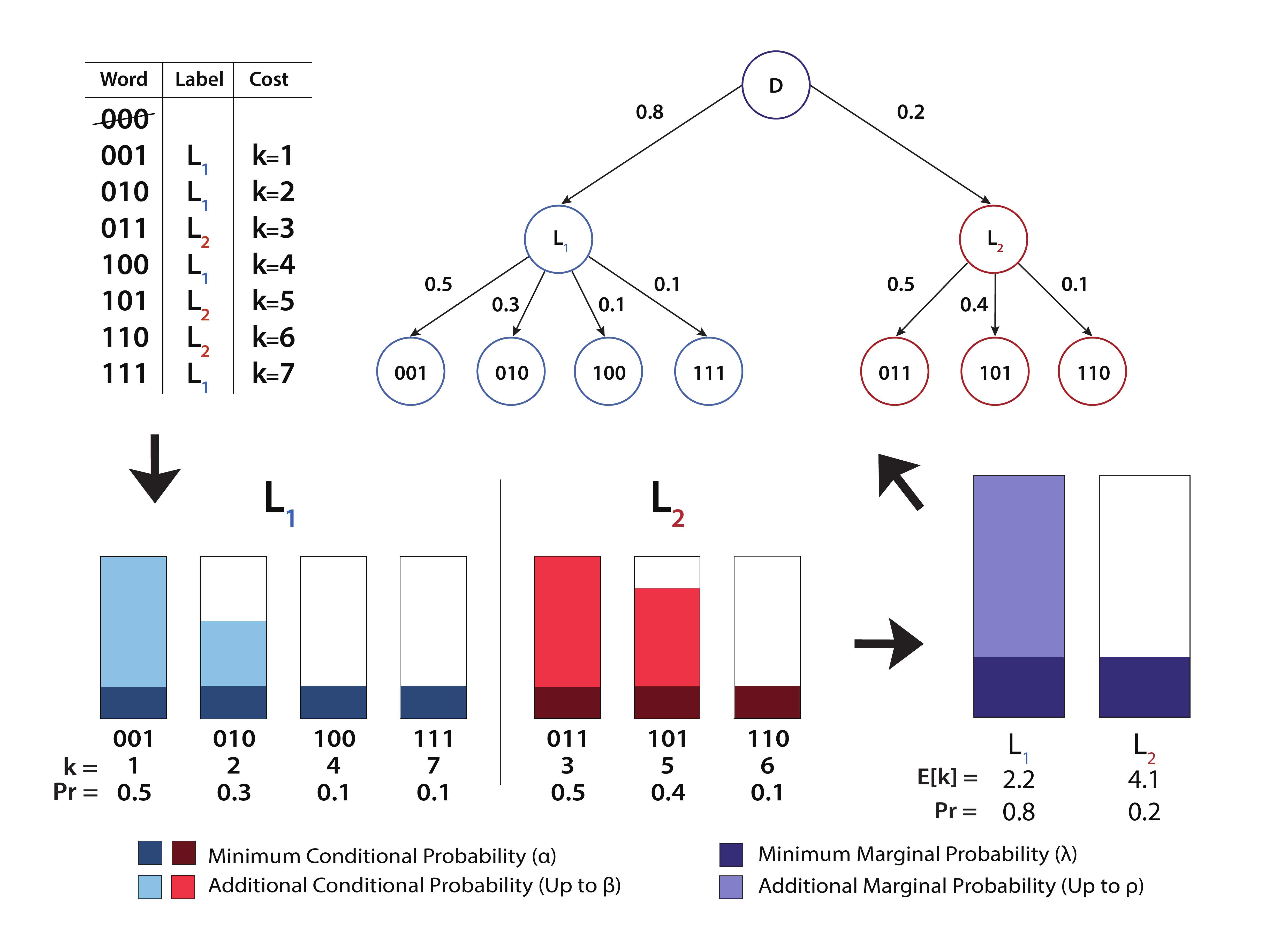}
    \caption{Applying the greedy LQCI construction to our toy example. Counter-clockwise from upper left: table of improvisations, the greedy cost construction, the greedy label construction, and the final improvising distribution.}
    \label{fig:greedy_construction}
\end{figure}

We now describe the two phases of our construction formally.

\noindent\textbf{\textit{The Greedy Cost Construction.}}
For a particular label class $i \in \{1, \hdots, |\Omega| \}$, we proceed as follows. Let $\delta^i = (\delta^i_1, \dots, \delta^i_{|\Theta|})$ be a list of all the cost classes $I_{i,k}$ with label $i$, sorted in increasing order of cost. Then fix $o_i = \frac{1 - \hat{\alpha}_i |I_i|}{\hat{\beta}_i - \hat{\alpha}_i}$, whose floor is the maximum number of words that can be assigned $\hat{\beta}_i$ probability (the maximum allowed) while still leaving at least $\hat{\alpha}_i$ probability (the minimum allowed) for each remaining word. Then, moving through the cost classes in the order given by $\delta^i$, we assign $\hat{\beta}_i$ probability to each word in the class, until we get to a class $\delta^i_r$ where the cumulative number of words so far (including the new class) would exceed $o_i$. To this class we assign $\hat{\beta}_i(o_i - \sum_{k=1}^{r-1} |\delta_k^i|) + \hat{\alpha}_i (\sum_{k=1}^r |\delta_k^i| - o_i)$ probability (spread uniformly over words in the class), the maximum allowed while leaving exactly $\hat{\alpha}_i$ for each remaining word.
Assigning $\hat{\alpha}_i$ to the remaining words, we obtain a distribution $D_i$ over the whole label class $I_i$.

We note that this process is not well-defined when $\hat{\alpha}_i =\hat{\beta}_i$ (in which case we simply assign probability $\hat{\alpha}_i$ to every word in $I_i$) or when $\hat{\alpha}_i |I_i| > 1$ (in which case the instance is infeasible due to $\hat{\alpha}_i$ being too large); also, the process does not result in a probability distribution if $\hat{\beta}_i |I_i| < 1$ (in which case the instance is infeasible due to $\hat{\beta}_i$ being too small).
Except in these cases, we get a well-defined distribution $D_i$ over $I_i$ which satisfies conditions (1) and (4) of Def.~\ref{def:improvising_dist}.
Moreover, the expected cost of $D_i$ is minimal among all such distributions, since it assigns as much weight as possible to the words with lowest cost.

\medskip

\noindent\textbf{\textit{The Greedy Label Construction.}}
Given the distributions $D_i$ for each label class $I_i$ from the first stage, we now choose a distribution over labels.
Following a similar pattern as before, let $\delta$ be a list of the distributions $D_i$ sorted in order of increasing expected cost.
Then fix $u = \lfloor \frac{1 - |\Omega| \lambda}{\rho - \lambda} \rfloor$, which is the number of label classes that can be assigned probability $\rho$ (the maximum allowed) while still leaving at least $\lambda$ (the minimum allowed) for each remaining class.
We assign $\rho$ probability to the first $u$ label classes in $\delta$.
To the next label class we assign probability $1 - \rho u  - \lambda (|\Omega| - u - 1)$, the maximum allowed while leaving exactly $\lambda$ for each remaining label class.
Finally, we assign $\lambda$ to all remaining label classes, and call the resulting distribution over labels $\hat{D}$.
Similar to before, this process will be well-defined and result in a distribution when $\frac{1}{\rho} \leq |\Omega| \leq \frac{1}{\lambda}$; otherwise, $\rho$ is too small or $\lambda$ is too large for condition (3) of Def.~\ref{def:improvising_dist} to be satisfied.

\medskip

To complete the construction, we obtain a final distribution $D$ over words by first sampling a label $i$ from $\hat{D}$ and then sampling from $D_i$.
The greedy cost construction ensured that $D_i$ is defined over the class $I_i \subseteq I$ and assigns probability between $\hat{\alpha}_i$ and $\hat{\beta}_i$ to each word, so $D$ will satisfy the hard and randomness over words constraints in Def.~\ref{def:improvising_dist}.
The greedy label construction ensures that $\hat{D}$ assigns probability between $\lambda$ and $\rho$ to each label, so $D$ will also satisfy the randomness over labels constraint.
Finally, since each phase selects a distribution of minimal cost amongst those satisfying the corresponding constraints, if \emph{any} improvising distribution exists then $D$ will have no greater cost, thereby satisfying the cost constraint and being an improvising distribution.
Formalizing this argument yields the following theorem (see the Appendix for a detailed proof):
		
	\begin{theorem}
	    \label{thm:lqci_greedy}
		An LQCI instance is feasible if and only if all of the following conditions are true:
		
		\begin{enumerate}
			\item $ \displaystyle \frac{1}{\rho} \leq |\Omega| \leq \frac{1}{\lambda}$
			\item $ \displaystyle \forall i \in \{1, \dots, |\Omega|\}, \ \frac{1}{\hat{\beta_i}} \leq |I_i| \leq \frac{1}{\hat{\alpha_i}}$
			\item The greedy LQCI construction produces a distribution $D$ whose expected cost is at most $c$ (i.e., $E[\CF(w) \mid w \leftarrow D] \leq c$).
		\end{enumerate}
	\end{theorem}

We conclude this section with a reminder that the greedy LQCI construction is a \emph{construction} and not a practical algorithm: it defines a distribution but not a practical way to compute it for a specified LQCI instance.
With common specification formalisms such as DFAs and Boolean formulas, the number of possible improvisations can easily be exponential in the size of the problem instance.
In this case, assigning probabilities to words one at a time as described above in the abstract construction would be highly impractical. Instead, the algorithms we present in the following sections are able to avoid enumerating exponentially-large sets by working with implicit representations to create distributions equal to or approximating the one produced by the greedy LQCI construction.

\section{Exact LQCI for Automata Specifications}
\label{sec:exact-schemes}

The greedy LQCI construction from Sec.~\ref{sec:theory} gives us a way to determine if an LQCI instance is feasible and, if so, to build an improvising distribution.
Implementing the construction requires several operations---such as computing the size of the label/cost classes---which may or may not be tractable depending on the types of specification used in the instance.
In this section, we will identify a sufficient list of operations which yield an efficient generic improvisation scheme for any class of LQCI instances with specifications supporting these operations.
Then we will instantiate the scheme for two natural classes of specifications given by deterministic finite automata, obtaining efficient improvisation algorithms.

Following the description of the preceding section, we can see that for a given LQCI instance, the operations listed below are sufficient to complete the greedy LQCI construction and sample from the resulting distribution:

\begin{definition}{(Sufficient Operations)}
\label{def:sufficient_ops}
    Given an LQCI instance $\mathcal{C}$:
    \begin{enumerate}
        \item \label{op:compute_costs} Compute the list of possible costs $\Theta$. 
        \item \label{op:count_cost_class} For each $i \in \{1, \hdots, |\Omega|\}$ and $k \in \Theta$, compute $|I_{i,k}|$.
        \item \label{op:sample_cost_class} For each $i \in \{1, \hdots, |\Omega|\}$ and $k \in \Theta$, sample uniformly from $I_{i,k}$.
    \end{enumerate}
\end{definition}

If we can implement these operations in polynomial time, we can build a polynomial-time improvisation scheme in the sense of \cite{fremont:ci_original,fremont:ci_extended}, i.e., an algorithm which solves the LQCI problem in polynomial time, and whose generated improvisers themselves run in polynomial (expected) time. To do this we first compute the list of possible costs and the size of each $I_{i,k}$. We then perform a modified version of the greedy construction which assigns probabilities to entire cost classes instead of individual words. As each word in a class has the same label and cost, we can satisfy our cost and randomness requirements with a distribution that assigns the same probability to every word within a class.
Then to implement placing probability $p$ on each word of $I_{i,k}$ without enumerating this potentially exponentially-large set, we simply choose the set with probability $p |I_{i,k}|$ and then sample uniformly from it (see the Appendix for a detailed argument).


\begin{theorem}
\label{thm:exact_scheme}
Suppose for a class of LQCI instances the operations in Def.~\ref{def:sufficient_ops} can be performed in polynomial time (in the size of the instance). Then there is a polynomial-time improvisation scheme for that class.
\end{theorem}

One broad class of specifications to which this scheme can apply is deterministic finite automata (DFAs): for example, we can encode the specifications from our robotic planning example as DFAs.
While a DFA can encode the hard specification $\HC$ directly, encoding cost and label functions is not as clear.
We consider two natural encodings: most simply, we can label each state of the DFA with an integer, assigning the associated label/cost to words ending at that state.

\begin{theorem}
\label{thm:output_dfa}
Consider the class of LQCI instances where $\HC$ is a DFA, $\CF$ and $\LF$ are given by DFAs which output an integer cost/label associated with the state they end on, the length bounds are given in unary and all other numerical parameters in binary.
This class has a polynomial-time improvisation scheme.
\end{theorem}
\begin{proof}[Sketch]
Operation (\ref{op:compute_costs}) is trivial.
For (\ref{op:count_cost_class}) and (\ref{op:sample_cost_class}), we can easily construct DFAs accepting all improvisations with a given label and cost, then apply classical techniques for counting/sampling from the language of a DFA~\cite{hickey:uniform_random_gen_cfg}.
\qed
\end{proof}

To capture cost functions like path length or mission time (as in our planning example), we consider a second encoding using weighted DFAs: states are again labeled with integers, but the cost is now given by \emph{accumulating} costs from every state passed through.
Here, the number of possible costs can grow linearly with the largest cost of a single state, and so be exponential in the size of the (binary) encoding; as a result we only obtain a \emph{pseudopolynomial} improvisation scheme by applying Thm.~\ref{thm:exact_scheme}.
The algorithm can still be feasible, however, when the magnitude of possible costs is not too large, as we will see in Sec.~\ref{sec:experiments}.

\begin{theorem}
\label{thm:acc_dfa}
Consider the class of LQCI instances as in Thm.~\ref{thm:output_dfa} but where $\CF$ is given by a weighted DFA, i.e. summing the integer costs associated with each state of a DFA accepting path (with multiplicity). 
This class has a pseudopolynomial improvisation scheme.
\end{theorem}
\begin{proof}[Sketch]
We can perform operation (\ref{op:compute_costs}) by dynamic programming over the states and word lengths up to the length bound $n$.
If the maximum cost of a state in the DFA for $\CF$ is $M$, then the cost of an improvisation is at most $M(n+1)$; so for (\ref{op:count_cost_class}) and (\ref{op:sample_cost_class}) we can build DFAs of size $\textit{poly}(M,n)$ recognizing $I_{i,k}$ and then apply counting/sampling as above.
If state costs were encoded in unary, the operations above would take polynomial time and Thm.~\ref{thm:exact_scheme} would apply.
Converting from binary to unary yields a pseudopolynomial scheme.
\qed
\end{proof}

\section{Approximate LQCI for Symbolic Specifications}
\label{sec:approx_lqci_alg}

The LQCI algorithms for DFAs that we developed in the previous section cover many useful specifications; however, as we will see in Sec.~\ref{sec:experiments}, even fairly simple specifications can require very large automata when represented explicitly.
In this section we propose an algorithm that avoids such blowup by working with \emph{symbolic} specifications given by Boolean formulas.
We cannot use our scheme of Thm.~\ref{thm:exact_scheme} directly, because counting the number of solutions of a Boolean formula is \textsf{\#P}-hard.
Nevertheless, we will show that by leveraging recent advances in SAT solving, we can \emph{approximately} solve LQCI to any desired accuracy.

We consider LQCI instances with specifications given by Boolean formulas, whose variables encode traces and costs; for modeling convenience, we also allow a vector of auxiliary variables $z$.
Specifically, we assume we are given:
\begin{itemize}
    \item a conjunctive normal form (CNF) formula $h(x, z)$ such that $\exists z. h(x, z)$ holds if and only if the bitvector $x$ encodes a trace satisfying the hard constraint;
    \item a CNF formula $\ell(x,y,z)$ such that $\exists z . \ell(x,y,z)$ holds if and only if trace $x$ has the label encoded by the bitvector $y$;
    \item a CNF formula $k(x,y,z)$ such that $\exists z . k(x,y,z)$ holds iff trace $x$ has cost $y$ (a positive integer).
\end{itemize}
We further assume that the instance has only a polynomial number of labels, although there can be exponentially-many costs.

Given such an instance, we can readily build a CNF formula $\phi_i (x,y,z)$ which is satisfiable iff $x$ encodes a word which has length between $m$ and $n$, satisfies the hard constraint, belongs to label $i$, and has cost $y$. The solutions $x$ for a particular choice of $i$ and $y$ comprise the associated cost class, so that the operations we need for the greedy construction are instances of the \emph{model counting} and \emph{uniform generation} problems for SAT.\footnote{Since we do not want to count over the auxiliary variables $z$, we actually require \emph{projected} counting/sampling, which the algorithms we use can also perform~\cite{unigen_14,fremont:ci_extended}.}
Recent work has yielded practical algorithms based on SAT solvers which solve these problems approximately~\cite{unigen_14,approxmc_unigen_20}\footnote{\label{note:unigen_proviso}We note that \texttt{UniGen}~\cite{unigen_14,unigen_15} is not strictly speaking an almost-uniform generator as in Def.~\ref{def:approx_counter} since it only supports sufficiently-large tolerances; for theoretical results, one can substitute the algorithm of \cite{bgp_sampling} to do \emph{exact} (projected) uniform sampling.}:

\begin{definition}{(\cite{unigen_14})}
\label{def:approx_counter}
An \emph{approximate counter} is a probabilistic algorithm $\mathcal{C}$ which given a CNF formula $F$ with set of solutions $R_F$, a tolerance $\tau > 0$, and a confidence $1-\delta \in [0,1)$ guarantees that \[ \Pr \left[ |R_F|/(1+\tau) \leq \mathcal{C}(F,\tau,1-\delta) \leq (1+\tau)|R_F| \right] \geq 1 - \delta . \]
An \emph{almost-uniform generator} $\mathcal{G}$ is a probabilistic algorithm that, given $F$ as above and a tolerance $\epsilon > 0$, guarantees that for every $y \in R_F$, we have \[ 1/((1+\epsilon)|R_F|) \leq \Pr[\mathcal{G}(F,\epsilon) = y] \leq (1+\epsilon)/|R_F| .\]
\end{definition}

We can modify our greedy construction to work with only approximate counting/sampling as follows.
If the cost bitvector has $|y|$ bits, the cost of a word is between $1$ and $2^{|y|}$.
To avoid enumerating exponentially-many cost classes for label $i$, we group words into ``cost buckets'' by subdividing this interval into powers of $r$ for some $r>1$, i.e. $[1,r), [r, r^2), \dots, [r^{b-1}, r^b)$.
We will have $b = O(\log_r (2^{|y|})) = O(|y| / \log r)$ buckets, and we can estimate the size of bucket $j$ by approximately counting solutions to $\exists z . [ \phi_i(x,y,z) \land (r^j \leq y < r^{j+1})]$.
We will then use these estimates to choose a distribution over buckets, following the intuition of the greedy cost construction that we should assign the most probability to buckets with lowest estimated cost, but with some adjustments to bound the error that approximate sampling introduces.


For each label class $i$ with randomness parameters $\alpha$ and $\beta$, we apply a modified form of the greedy cost construction, shown in Algorithm~\ref{alg:approx_greedy_cost}. 
We start in lines 1-3 by using model counting as above (with a tolerance $\tau$ and confidence $1-\delta$ to be specified later) to find estimates $c_k$ of the size of each bucket $k$, and corresponding lower bounds $p_k$ on how much probability the bucket would have received in the exact greedy construction (the extra $1+\tau$ factor accounting for possibly overestimating the size of the bucket).
If these lower bounds total more than 1, then we know there are too many improvisations for the instance to be feasible (assuming the model counts are within their tolerance) and we return false on line 4.
Otherwise, on lines 5-7 we proceed as in the greedy construction, starting from the cheapest bucket, increasing the assigned probability per word to $(1+\tau)\beta$ until a probability of 1 is reached.
The factor of $1+\tau$ ensures that, even if the model counts have underestimated the size of the cheaper buckets, we still assign them at least as much probability as the exact greedy construction would.
Next, line 8 checks if there are too few improvisations, similarly to line 4.
Finally, we return our distribution over buckets, as well as a lower bound on its expected cost that we will use next.

\begin{algorithm}[tb]
\caption{ApproximateGreedyCost($i,\alpha,\beta,r,b,\tau,\delta$)}
\label{alg:approx_greedy_cost}
\begin{algorithmic}[1]
	\For{$k = 1 \text{ to } b$}
		\State $c_k := \#SAT(\exists z. \phi_i(x, y, z) \land (r^{k-1} \leq y < r^k), \tau, 1 - \delta)$
		\State $p_k := \alpha c_k / (1+\tau)$
	\EndFor
	\lIf{$\sum_{j=1}^{b} p_j > 1$} $\textbf{return} \text{ False}$
	\For{$k = 1 \text{ to } b$}
		\State $p_k := \min((1+\tau)\beta c_k, 1 - \sum_{j\neq k} p_j)$
		\lIf{$\sum_{j=1}^{b} p_j = 1$} $\textbf{break}$
	\EndFor
	\lIf{$\sum_{j=1}^{b} p_j < 1$} $\textbf{return} \text{ False}$
	\State $Lo := \sum_{j=1}^{b} p_j r^{j-1}$
	\State $\textbf{return } \{p_j\}_{j=1}^{b}, Lo$ \label{line:agc_feasible}
\end{algorithmic}
\end{algorithm}

If Algorithm~\ref{alg:approx_greedy_cost} does not return false for any label class, then we complete our approximate LQCI algorithm by running the greedy label construction from Sec.~\ref{sec:theory}, using the lower bounds from Algorithm 1 as the expected cost of each label class.
As before, we declare the instance infeasible if the construction fails or if its expected cost exceeds the cost bound $c$.
Otherwise, we obtain a distribution over all the cost buckets; our improviser then simply chooses a bucket from this distribution and applies almost-uniform sampling to sample a word from it.

Choosing the bucket count and counting/sampling tolerances appropriately, our algorithm can approximate an improvising distribution to within arbitrarily-small multiplicative error, using polynomially-many calls to a SAT solver:

\begin{theorem}
\label{thm:approx_lqci}
There is an algorithm which, given a Boolean LQCI instance $\mathcal{C}$, a cost tolerance $\zeta > 0$, a randomness tolerance $\gamma > 0$, and a confidence $1-\delta \in [0, 1)$, runs in $\textit{poly}(|\mathcal{C}|, 1/\zeta, 1/\gamma, \log(1/\delta))$ time relative to an \textsf{NP} oracle and either returns $\bot$ or an algorithm sampling from a distribution $\widetilde{D}$ over words.
With probability at least $1-\delta$, if $\bot$ is returned then $\mathcal{C}$ is infeasible, and otherwise:
\begin{enumerate}
	\item \textbf{Hard Constraint:} $ \displaystyle \Pr[\HC(w) \mid w \leftarrow \widetilde{D}] = 1$
			
	\item \textbf{Cost Constraint:} $ \displaystyle E[\CF(w) \mid w \leftarrow \widetilde{D}] \leq (1+\zeta) \cost$ 
		
	\item \textbf{Randomness over Labels:} $ \displaystyle \forall i \in \{1, \hdots, |\Omega| \}, \ \lambda \leq \Pr{[w \in I_i \mid w \leftarrow \widetilde{D}]} \leq \rho $
			
	\item \textbf{Randomness over Words:} $ \displaystyle \forall i \in \{1, \hdots, |\Omega| \} \ \forall y \in I_i, \newline  \hat{\alpha}_i/(1+\gamma) \leq \Pr[y = w \mid w \in I_i, w \leftarrow \widetilde{D}] \leq (1+\gamma) \hat{\beta}_i$
\end{enumerate}
\end{theorem}

\section{Maximum-Entropy LQCI}
\label{sec:maxent}

Our LQCI definition requires providing conditional probability bounds for every label, which while allowing maximal control of the distribution, can be unwieldy to use.
However, if we drop conditional bounds entirely, trivial solutions with unnecessarily-poor randomness can appear.
For example, consider an LQCI instance with parameters $\lambda = 0.5$, $\rho = 0.5$, $\hat{\alpha} = (0,\hdots,0)$, $\hat{\beta} = (1,\hdots,1)$. With this choice, any distribution will satisfy the randomness over words constraint, and all labels have the same marginal probability of being selected. Then assume that we have two labels, costs $\Theta = (1,2)$, and cost bound $c = 1.5$, along with the following cost class sizes: $|I_{1,1}| = 1$, $|I_{2,1}| = 1$, $|I_{1,2}| = 1000$, $|I_{2,2}| = 1000$.
Now simply assigning 50\% probability to $I_{1,1}$ and 50\% probability to $I_{2,1}$ is an improvising distribution.
Assigning 25\% probability to all 4 classes is also an improvising distribution, and clearly preferable from the perspective of randomness.
Unfortunately, without a nontrivial randomness over words constraint, we have no way to push the improviser to select the second distribution. To enforce this, we introduce the concept of entropy from information theory.

    \begin{definition}
    Given a discrete random variable $X$ with a set of outcomes $\Omega$ and probabilities $p: \Omega \rightarrow [0,1]$, the \emph{entropy} of $X$ is $H(X) = - \sum_{x \in \Omega} p(x) \lg p(x)$.
    \end{definition}

To obtain a problem formulation that maximizes randomness without requiring probability bounds for each class, we invoke the Principle of Maximum Entropy: amongst all improvising distributions (without a randomness over words constraint), we should select the one with the highest entropy (as first proposed for reactive CI in \cite{vazquezchanlatte:maximum_entropy_ci}). This yields a notion of Maximum-Entropy LQCI:

	\begin{definition}
	A \emph{Maximum-Entropy LQCI (MELQCI) instance} is an LQCI instance where $\hat{\alpha} = (0,\hdots,0)$ and $\hat{\beta} = (1,\hdots,1)$.
	A \emph{$\tau$-improviser} for a MELQCI instance $\mathcal{C}$ is an improviser (as in LQCI) whose output distribution has entropy at most $\tau$ less than the maximum-entropy improvising distribution for $\mathcal{C}$.
	We define the \emph{MELQCI problem} as, given an instance $\mathcal{C}$ and $\tau > 0$, determining if $\CII$ is feasible, and, if so, generating a $\tau$-improviser for $\CII$.
	\end{definition}

	We can solve MELQCI efficiently in the same cases as LQCI:
	
	\begin{theorem}
	    \label{thm:melqci_algorithm}
	    Given a class of MELQCI instances for which one can perform the operations in Def.~\ref{def:sufficient_ops} in polynomial time, there is a polynomial-time algorithm which given an instance from the class and a $\tau > 0$, computes a $\tau$-improviser.
	\end{theorem}
	\begin{proof}(Sketch).
	    Once cost class sizes have been computed as in Thm.~\ref{thm:exact_scheme}, the search for the desired distribution over cost classes can be formulated as an optimization problem with a separable convex objective (the entropy of the distribution) and linear constraints (improviser constraints). This problem can be solved in time polynomial in the size of the instance and $\log(1/\tau)$~\cite{chubanov:convex_opt_poly_alg}.
	\end{proof}

    As in Sec.~\ref{sec:exact-schemes}, we can transform this algorithm into a \emph{pseudopolynomial} scheme for accumulated-cost DFA specifications.

\section{Experiments}
\label{sec:experiments}

We ran several experiments on the robotic planning problems from Sec.~\ref{sec:prob_def} (code available at \cite{artifact}).
These experiments aim to demonstrate that we can encode practical problems as LQCI instances solvable using our algorithms, highlight the relative advantages/disadvantages of our exact/approximate algorithms, and show the necessity of the label function in ensuring meaningful randomness.

As a minimal experiment, we used a 6x6 grid world with a small range of costs (0--3 per cell, 8--39 for paths); we compared against a 7x7 grid world with a much larger range of costs (0--9 per cell, 38--137 for paths).\footnote{A larger 8x8 map exceeded our 24-hour wallclock timeout for all exact and approximate experiments.}
We encoded the specifications in Sec.~\ref{sec:prob_def} both as DFAs for our exact LQCI and MELQCI algorithms, and as Boolean formulas for our approximate LQCI algorithm. 
The Boolean encodings were obtained by formulating the specifications in the SMT theory of bitvectors, and bit-blasting them with \texttt{Z3}~\cite{z3_paper}; the resulting formulas had several thousand variables and tens of thousands of clauses.
We used \texttt{UniGen3}~\cite{approxmc_unigen_20,unigen_14} for uniform generation with its default tolerance\footnote{\texttt{UniGen3} cannot guarantee a multiplicative error of less than 7.48 \cite{unigen_15}; see footnote~\ref{note:unigen_proviso}.} of 17, and an in-development version of \texttt{ApproxMC}~\cite{approxmc_unigen_20,approxmc_16,arjun} for approximate model counting with tolerances of 1.4, 6.7, and 23.25, so that the overall $\gamma$ values were $10^2$, $10^3$, and $10^4$. To put these values into context, the small/large maps had on the order of $10^7$/$10^9$ improvisations, and we required that no word have $> \rho = 10^{-5}$ probability of being selected. Therefore, with our tightest/loosest $\gamma$ we are guaranteed that no word will be more than 0.1\%/10\% of the distribution respectively.
The confidence was set to 0.8 ($\delta = 0.2$), ApproxMC's default confidence. Each model counting call however required a much higher confidence to achieve an overall $\delta$ of 0.2.

For the small/large maps respectively we used length bounds of (1,25)/(1,30) and cost bounds of 30/50.
We used label probability bounds of $(0.3,0.4)$ throughout, except for unlabeled ``QCI'' experiments.
The experiments were run on a 64-core machine with 188 GB of RAM; we used 62 parallel threads, unless this exhausted RAM, in which case we used 16 threads.
The experiments are summarized in Tab.~\ref{tab:experiments}; due to significant runtime variability for the approximate experiments, we report means and standard deviations over 10 repetitions.
For all exact experiments which completed within the 24-hour wallclock timeout, RAM usage was $\le 6$ GB per thread, and the average time to sample an improvisation was $\le 1$ ms; all approximate experiments required $\le 250$ MB RAM per thread and took $\sim 20$ s to sample an improvisation.

\begin{table}[tb!]
\caption{Experiment parameters and improviser construction times (in minutes).}
\centering
{
\setlength{\tabcolsep}{3pt}
\begin{tabular}{ccccccccc}
\toprule
Map & Problem Type & $(\lambda, \rho)$ & $(\hat{\alpha}_i, \hat{\beta}_i)$ & $r$ & $\gamma$ & $\delta$ & Wall Time & CPU Time \\
\midrule
\multirow{6}{*}{6x6} & Exact QCI & (0, 3e-5) & (1, 1) & \multicolumn{3}{c}{\multirow{3}{*}{N/A}} &  540 & 5568 \\
 & Exact LQCI & (0, 1e-5) & (0.3, 0.4) & \multicolumn{3}{c}{} &  444 & 6102 \\
  & Exact MELQCI & N/A & (0.3, 0.4) & \multicolumn{3}{c}{} &  444 \tablefootnote{The LQCI/MELQCI runtimes were nearly identical, since MELQCI reuses the LQCI computations and adds a convex optimization step, which took negligible time.} & 6102 \\ \cmidrule{2-9} 
  & Approx. LQCI & (0, 1e-5) & (0.3, 0.4) & 1.2 & $10^2$ & 0.2 & $23.7 \pm 0.6$ & $93.3  \pm 1.4$ \\ 
  & Approx. LQCI & (0, 1e-5) & (0.3, 0.4) & 1.2 & $10^3$ & 0.2 & $21.2 \pm 0.7$ & $81.5  \pm 1.1$ \\
  & Approx. LQCI & (0, 1e-5) & (0.3, 0.4) & 1.2 & $10^4$ & 0.2 & $20.2 \pm 0.7$ & $78.4  \pm 3.4$  \\
\midrule
 \multirow{6}{*}{7x7} & Exact QCI & (0, 3e-5) & (1, 1) & \multicolumn{3}{c}{\multirow{3}{*}{N/A}} & \multicolumn{2}{c}{\multirowcell{3}{Timed out\\(24-hour wall time)}} \\
  & Exact LQCI & (0, 1e-5) & (0.3, 0.4) & \multicolumn{3}{l}{} & \multicolumn{2}{c}{} \\
  & Exact MELQCI & N/A & (0.3, 0.4) & \multicolumn{3}{l}{} & \multicolumn{2}{c}{} \\ \cmidrule{2-9} 
  & Approx. LQCI & (0, 1e-5) & (0.3, 0.4) & 1.2 & $10^2$ & 0.2 & $42.8 \pm 2.1$ & $186.1 \pm 3.9$ \\
  & Approx. LQCI & (0, 1e-5) & (0.3, 0.4) & 1.2 & $10^3$ & 0.2 & $38.8 \pm 8.8$ & $152.6 \pm 9.0$ \\
  & Approx. LQCI & (0, 1e-5) & (0.3, 0.4) & 1.2 & $10^4$ & 0.2 & $38.8 \pm 9.7$ & $145.5 \pm 9.5$ \\
\bottomrule
\end{tabular}
}
\label{tab:experiments}
\end{table}

We can draw several conclusions from these results. Improviser construction with the exact algorithm is significantly more expensive than with the approximate algorithm, in both CPU time and RAM. This is not surprising, as the exact encodings resulted in enormous DFAs which, for the large map, approached $10^{10}$ states. Conversely, sampling is much faster for the exact algorithm, with no SAT queries required. We can also see that the approximate algorithm can be used to practically solve problems that are infeasible to solve exactly, such as the large-map problem. We expect new developments in the relatively young field of approximate model counting/sampling will further speed up our algorithm.

Visualizing several randomly-chosen traces from our exact QCI and MELQCI experiments in Fig. \ref{fig:rp_traces}, we can see the importance of labels. In unlabeled QCI, the robot always charged at the substation near the main road due to the lower expected cost of such paths. In contrast, MELQCI yielded a near-uniform distribution over the charging stations. This increase in diversity was not free, with the average cost rising to 21.4 for MELQCI from 8.7 for QCI.
This trade-off demonstrates how LQCI allows us to balance the need for control over our improvisations with the need for meaningful diversity (not merely randomness) by choosing appropriate label functions.

\begin{figure}[tb]
    \centering
    \begin{subfigure}[b]{0.35\textwidth}
        \centering
        \includegraphics[width=\textwidth]{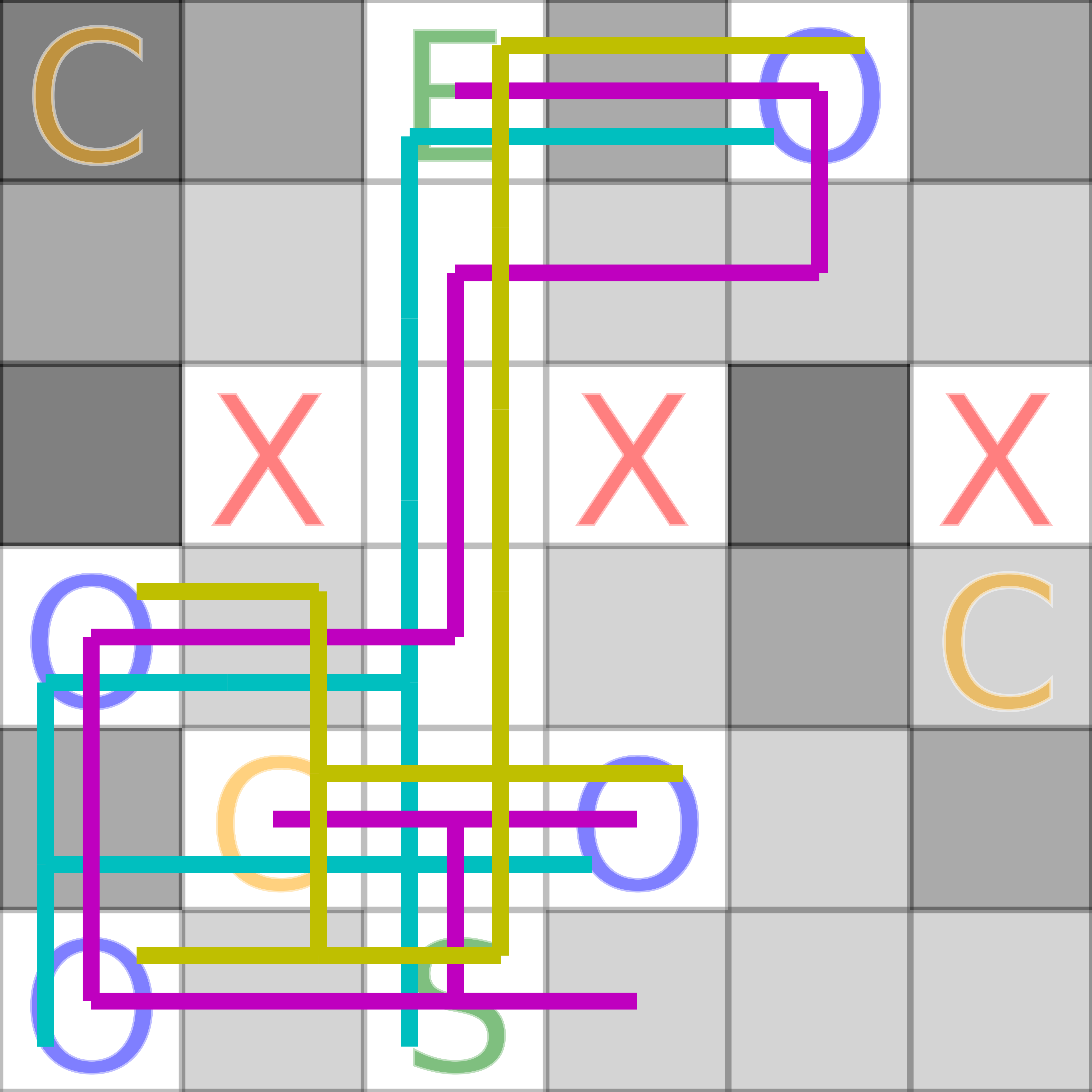}
        \caption{QCI Traces}
    \end{subfigure}
    \hspace{0.1\textwidth}
    \begin{subfigure}[b]{0.35\textwidth}
        \centering
        \includegraphics[width=\textwidth]{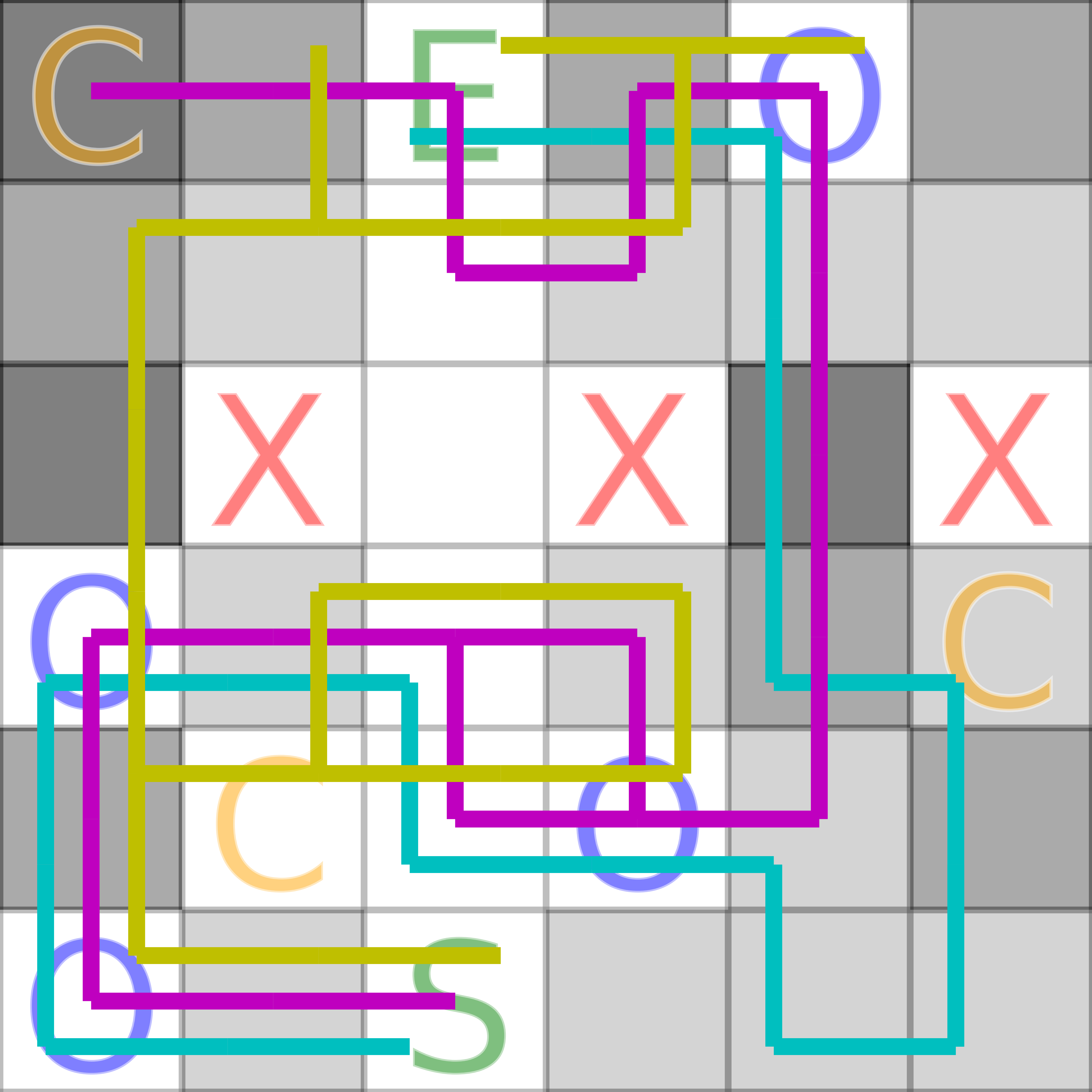}
        \caption{MELQCI Traces}
    \end{subfigure}
    \caption{Randomly-selected traces generated by the QCI/MELQCI improvisers for the 6x6 map. Note that all the QCI traces use the same charging station.}
    \label{fig:rp_traces}
\end{figure}

\section{Conclusion}
\label{sec:conclusion}

In this paper, we introduced \emph{labelled quantitative control improvisation} as a framework allowing correct-by-construction synthesis of randomized systems whose behavior must be diverse with respect to a label function and near-optimal with respect to a cost function.
We studied the theory of LQCI problems and developed algorithms for solving them for broad classes of specifications encoded as finite automata or Boolean formulas.
Our experiments demonstrated how our framework can be used to formalize and solve realistic robotic planning problems.

There are a number of clear directions for future work.
Scalability is an evident concern: our experiments show that our algorithms can require substantial resources to solve even relatively small LQCI problems.
While LQCI with Boolean formulas is a difficult \textsf{\#P}-hard problem, our algorithms will directly benefit from future progress in model counting; our DFA algorithms could also be improved through the use of abstraction to reduce state-space explosion.
We also plan to explore generalizations of our algorithms, such as extending our approximate scheme to MELQCI and to problems with exponentially-many labels, as well as potentially infinite traces.
Finally, we are investigating extensions of the LQCI problem to \emph{reactive} settings with adversarial environments, and to \emph{black-box} settings for design-space exploration and other problems where we do not have complete models for the cost function and other constraints.\\

\noindent\textbf{Acknowledgements.}
The authors thank Skyler Stewart for designing Fig. 2, and several anonymous reviewers for their helpful comments.
This work was supported in part by DARPA contract FA8750-20-C-0156 (SDCPS).

\bibliographystyle{splncs04}
\bibliography{bibliography}

\newpage

\appendix

\section{Proofs}

\textbf{Theorem \ref{thm:lqci_greedy} Full Proof}
\begin{proof} 
First, we verify that when conditions (1) and (2) of the Theorem hold, the greedy LQCI construction defines a valid probability distribution.
Since probability of words under $D_i$ is uniform within a cost class, we will use an abuse of notation to apply $D_i$ to cost classes as well, to be understood as the sum of the probabilities of the words in that cost class. As described, if cost class $I_{i,k}$ has index $j$ in the sorted list $\delta^i$ (i.e., $I_{i,k} = \delta^i_j$), then our distribution $D_i$ over the cost classes for label $i$ assigns it probability as follows:

	\begin{align*}
			D_i(I_{i,k}) = 
			D_i(\delta^i_j) =	\begin{cases}
									\hat{\beta}_i |\delta^i_j| \qquad &j < r			\\
									\hat{\beta}_i(o_i - \sum_{k=1}^{r-1} |\delta_k^i|) + \hat{\alpha}_i (\sum_{k=1}^r |\delta_k^i| - o_i) \qquad &j = r	\\
									\hat{\alpha}_i |\delta^i_j| \qquad &j > r			\\
								\end{cases}
		\end{align*}
		(recalling that $r \ge 1$ is defined to be the least index such that $\sum_{k=1}^r |\delta^i_k| > o_i$, if one exists; otherwise we may set $r = |\Theta|+1$).
		
The distribution over label classes is given as follows, if label class $I_i$ has index $j$ in the sorted list $\delta$ (i.e. $I_i = \delta_j$):

		\begin{align*}
			\hat{D}(I_i) = \hat{D}(\delta_j) = 	\begin{cases}
							 							\rho 	& \qquad j \leq u		\\
							 							1 - \rho u  - \lambda (|\Omega| - u - 1)  	& \qquad j = u + 1		\\
														\lambda & \qquad j > u + 1
													\end{cases}
 		\end{align*}
 		
 		Since the sets $I_i$ partition $I$ and the sets $I_{i,k}$ partition $I_i$, we now construct a unified distribution over all our words by combining our marginal and conditional distributions, again with assumed uniform probability within the cost classes, as follows,
 		
 		\begin{align*}
 			D(I_{i,k}) = \hat{D}(I_i) D_i(I_{i,k})
 		\end{align*}
 		
 		However, for this to be well-defined we must show that $\hat{D}$ and each $D_i$ are valid probability distributions that sum to 1.
 		For each $D_i$, its probabilities are nonnegative by the definition of $r$ and the fact that $o_i = (1 - \hat{\alpha}_i |I_i|)/(\hat{\beta}_i - \hat{\alpha}_i) \ge 0$ (due to our assumption that $|I_i| \le 1 / \hat{\alpha}_i$).
 		The sum of the probabilities is:
 		
 		\begin{align*}
 			\sum_{w \in I_{i,k}} D(w) &= \hat{\beta}_i \sum_{k=1}^{r-1} |\delta_k^i| + \hat{\beta}_i(o_i - \sum_{k=1}^{r-1} |\delta_k^i|) + \hat{\alpha}_i (\sum_{k=1}^r |\delta_k^i| - o_i) + \hat{\alpha}_i \sum_{k=r+1}^{|\Theta|} |\delta_k^i|	\\
 			&= \hat{\beta}_i o_i + \hat{\alpha}_i \sum_{k=1}^r |\delta_k^i| - \hat{\alpha}_i o_i + \hat{\alpha}_i \sum_{k=r+1}^{|\Theta|} |\delta_k^i|	\\
 			&= \hat{\beta}_i \frac{1 - \hat{\alpha}_i |I_i|}{\hat{\beta}_i - \hat{\alpha}_i} + \hat{\alpha}_i \sum_{k=1}^r |\delta_k^i| - \hat{\alpha}_i \frac{1 - \hat{\alpha}_i |I_i|}{\hat{\beta}_i - \hat{\alpha}_i} + \hat{\alpha}_i \sum_{k=r+1}^{|\Theta|} |\delta_k^i|	\\
 			&= 1 - \hat{\alpha}_i |I_i| + \hat{\alpha}_i \sum_{k=1}^r |\delta_k^i| + \hat{\alpha}_i \sum_{k=r+1}^{|\Theta|} |\delta_k^i|	\\
 			&= 1 - \hat{\alpha}_i |I_i| + \hat{\alpha}_i |I_i|	\\
 			&= 1
 		\end{align*}
 		as desired.
 		Note that in the case when $\hat{\alpha}_i = \hat{\beta}_i$, as stated above we instead define $D_i(\delta^i_j) = \hat{\alpha}_i |\delta^i_j|$: the probabilities are still nonnegative, and sum to 1 due to our assumption that $1 / \hat{\beta}_i \le |I_i| \le 1 / \hat{\alpha}_i$, which forces $|I_i| = 1 / \hat{\alpha}_i$.
 		
 		For $\hat{D}$, its probabilities are nonnegative since we assume $|\Omega| \le 1/\lambda$ and so
 		\begin{align*}
 		1 - \rho u - \lambda(|\Omega| - u - 1) &= 1 - (\rho - \lambda) u - \lambda |\Omega| + \lambda \\
 		&= 1 - \lambda |\Omega| - (\rho - \lambda) \floor{\frac{1 - \lambda|\Omega|}{\rho - \lambda}} + \lambda \\
 		&\ge \lambda \ge 0 .
 		\end{align*}
 		
 		The sum of the probabilities is:
		\begin{align*}
			\sum_{i = 1}^{|\Omega|} \hat{D}(I_i) &= \rho u + \left[ 1 - \rho u  - \lambda (|\Omega| - u - 1) \right] + \lambda(|\Omega| - u - 1)		\\
			&= \rho \floor{\frac{1 - |\Omega| \lambda}{\rho - \lambda}} + 1 + \lambda u - \rho u - \lambda |\Omega| + \lambda + \lambda(|\Omega| -  \floor{\frac{1 - |\Omega| \lambda}{\rho - \lambda}} - 1)	\\
			&= \rho \floor{\frac{1 - |\Omega| \lambda}{\rho - \lambda}} + 1 + \lambda u - \rho u - \lambda |\Omega| + \lambda + \lambda |\Omega| - \lambda \floor{\frac{1 - |\Omega| \lambda}{\rho - \lambda}} - \lambda	\\
			&= \rho \floor{\frac{1 - |\Omega| \lambda}{\rho - \lambda}} + 1 + \lambda u - \rho u - \lambda \floor{\frac{1 - |\Omega| \lambda}{\rho - \lambda}}	\\
			&= 1 + (\rho - \lambda) \floor{\frac{1 - |\Omega| \lambda}{\rho - \lambda}} + (\lambda - \rho) u	\\
			&= 1 + (\rho - \lambda) \floor{\frac{1 - |\Omega| \lambda}{\rho - \lambda}} + (\lambda - \rho) \floor{\frac{1 - |\Omega| \lambda}{\rho - \lambda}}\\
			&= 1
		\end{align*}
        again as desired.
        Note that in the case when $\lambda = \rho$, we instead define $\hat{D}(I_i) = \lambda$: the probabilities are still nonnegative, and sum to 1 due to our assumption that $1/\rho \le |\Omega| \le 1/\lambda$, which forces $|\Omega| = 1/\lambda$.

		As all our marginal and conditional distributions are well-defined probability distributions, it follows that $D$ is also a well-defined probability distribution.
		
Next we must show that a given LQCI instance is feasible if and only if all of the following hold:
		\begin{enumerate}
			\item $ \displaystyle \frac{1}{\rho} \leq |\Omega| \leq \frac{1}{\lambda}$
			\item $ \displaystyle \forall i \in \{1, \dots, |\Omega|\}, \ \frac{1}{\hat{\beta_i}} \leq |I_i| \leq \frac{1}{\hat{\alpha_i}}$
			\item The greedy LQCI construction produces a distribution $D$ whose expected cost satisfies $E[\mathcal{K}(w) \mid w \leftarrow D] \leq c$.
		\end{enumerate}
		
		\noindent \textbf{($\Leftarrow$) If the above conditions hold, the LQCI instance is feasible.} ~ \\
		    We showed above that conditions (1) and (2) imply the greedy LQCI construction defines a valid distribution $D$.
		    By construction, $D$ only gives probability to words contained in some cost class and therefore in $I$, so it satisfies the hard constraint.
		    Condition (3) states that $D$ satisfies the cost constraint.
		    By the definition of $\hat{D}$, the probability of each label class is between $\lambda$ and $\rho$, so $D$ satisfies the randomness over labels constraint.
		    Finally, the definition of $D_i$ ensures that the probability of each word in $I_i$ is between $\hat{\alpha}_i$ and $\hat{\beta}_i$, so $D$ satisfies the randomness over words constraint.
		    Therefore $D$ is an improvising distribution, so the LQCI instance is feasible.
		
		\noindent \textbf{($\Rightarrow$) If the LQCI instance is feasible, the above conditions hold.} ~ \\
			Assuming that the LQCI instance is feasible, and thus that there is an improvising distribution, we will show that the above conditions hold. Let $D'$ be the improvising distribution.
			
			\begin{enumerate}
				\item $ \displaystyle \frac{1}{\rho} \leq |\Omega| \leq \frac{1}{\lambda}$\\

				Note that a probability distribution must sum to 1, and that a label class must have between $\lambda$ and $\rho$ marginal probability of being selected in an improvising distribution, so,
				\begin{align*}
				1 &= \sum_{i = 1}^{|\Omega|} \Pr[w \in I_i \mid w \leftarrow D'] \leq \rho |\Omega|		\\
				\therefore \quad \frac{1}{\rho} &\leq |\Omega|						\\
				1 &= \sum_{i = 1}^{|\Omega|} \Pr[w \in I_i \mid w \leftarrow D'] \geq \lambda |\Omega|	\\
				\therefore \quad |\Omega| &\leq \frac{1}{\lambda}
				\end{align*}
				As we can see, $\frac{1}{\rho} \leq |\Omega| \leq \frac{1}{\lambda}$.
				
				\item $ \displaystyle \forall i \in \{1, \dots, |\Omega|\}, \ \frac{1}{\hat{\beta_i}} \leq |I_i| \leq \frac{1}{\hat{\alpha_i}}$\\

				Again, note that a probability distribution must sum to 1, and that a word must have between $\alpha_i$ and $\beta_i$ conditional probability of being selected, so,
				\begin{align*}
				1 &= \sum_{w \in I_i} \Pr[w \mid w \in I_i, w \leftarrow D'] \leq \hat{\beta}_i |I_i|						\\
				\therefore \quad \frac{1}{\hat{\beta}_i} &\leq |I_i|						\\
				1 &= \sum_{w \in I_i} \Pr[w \mid w \in I_i, w \leftarrow D'] \geq \hat{\alpha}_i |I_i|						\\
				\therefore \quad |I_i| &\leq \frac{1}{\hat{\alpha}_i}
				\end{align*}	
				As we can see, $\frac{1}{\hat{\beta}_i} \leq |I_i| \leq \frac{1}{\hat{\alpha}_i}$.
			
				\item The greedy LQCI construction satisfies the cost constraint.\\
				We will first show that the greedy LQCI construction satisfies all the non-cost LQCI constraints (hard constraint, randomness over words, and randomness over labels).
				To begin with $D$ samples exclusively over $I$, meaning that it trivially satisfies the hard constraint.
				
				We now show that $D$ satisfies the randomness over words constraint. In most cases this is trivially satisfied by having assigned $\hat{\alpha}_i$ or $\hat{\beta}_i$ probability to each word. The only exception is that for one cost class we assign $\hat{\beta}_i(o_i - \sum_{k=1}^{r-1} |\delta_k^i|) + \hat{\alpha}_i (\sum_{k=1}^r |\delta_k^i| - o_i)$ probability. However, below we show that we can bound this value to between $\hat{\alpha}_i |\delta^i_k|$ and $\hat{\beta}_i |\delta^i_k|$,

				\begin{align*}
					& \hat{\beta}_i(o_i - \sum_{k=1}^{r-1} |\delta_k^i|) + \hat{\alpha}_i (\sum_{k=1}^r |\delta_k^i| - o_i) 					\geq \hat{\alpha}_i(o_i - \sum_{k=1}^{r-1} |\delta_k^i|) + \hat{\alpha}_i (\sum_{k=1}^r |\delta_k^i| - o_i) 
					= \hat{\alpha}_i |\delta^i_r| \\
					& \hat{\beta}_i(o_i - \sum_{k=1}^{r-1} |\delta_k^i|) + \hat{\alpha}_i (\sum_{k=1}^r |\delta_k^i| - o_i) 					\leq \hat{\beta}_i(o_i - \sum_{k=1}^{r-1} |\delta_k^i|) + \hat{\beta}_i (\sum_{k=1}^r |\delta_k^i| - o_i) 
					= \hat{\beta}_i |\delta^i_r| \\
					\therefore \quad & \hat{\alpha}_i |\delta^i_r| \leq \hat{\beta}_i(o_i - \sum_{k=1}^{r-1} |\delta_k^i|) + \hat{\alpha}_i (\sum_{k=1}^r |\delta_k^i| - o_i) \leq \hat{\beta}_i |\delta^i_r|
				\end{align*}
				Therefore each word is also assigned a probability between $\hat{\alpha}_i$ and $\hat{\beta}_i$.
				
				We now show that $D$ satisfies the randomness over labels constraint. Again, in most cases this is trivially satisfied by assigning $\lambda$ or $\rho$ marginal probability to each label class. The only exception is that for one label class we assign $1 - \rho u - \lambda (|\Omega| - u - 1)$ marginal probability. Below we show that we can bound this value between $\lambda$ and $\rho$.
				
				 \begin{align*}
				 	& \ 1 + \lambda k - \rho k - |\Omega| \lambda + \lambda		\\
					&= 1 + \lambda \floor{\frac{1 - |\Omega| \lambda}{\rho - \lambda}} - \rho \floor{\frac{1 - |\Omega| \lambda}{\rho - \lambda}} - |\Omega| \lambda + \lambda								\\
					&= 1 + (\lambda - \rho) \floor{\frac{1 - |\Omega| \lambda}{\rho - \lambda}} - |\Omega| \lambda + \lambda	\quad (\text{Recall by definition $\lambda \leq \rho$}) \\	
				\end{align*}								
				
				Then for some $\iota <1:$
				
				\begin{minipage}{0.4\textwidth}
				\begin{align*}
					&\geq 1 + (\lambda - \rho) \frac{1 - |\Omega| \lambda}{\rho - \lambda} - |\Omega| \lambda + \lambda		\\
					&= 1 + (\lambda - \rho) \frac{|\Omega| \lambda - 1}{\lambda - \rho} - |\Omega| \lambda + \lambda		\\
					&= 1 + |\Omega| \lambda - 1  - |\Omega| \lambda + \lambda	\\
					&= \lambda			\\
					\therefore \quad &1 + \lambda k - \rho k - |\Omega| \lambda + \lambda \geq \lambda		\\
				\end{align*}
				\end{minipage}
				\begin{minipage}{0.4\textwidth}
				\begin{align*}
					&= 1 + (\lambda - \rho) \left( \frac{|\Omega| \lambda - 1}{\lambda - \rho} - \iota \right) - |\Omega| \lambda + \lambda	 \\ 
					&= 1 + |\Omega| \lambda - 1 - \iota \lambda + \epsilon \rho - |\Omega| \lambda + \lambda\\
					&= - \iota \lambda + \iota \rho + \lambda\\
					&= \iota (\rho - \lambda) + \lambda	\\
					&< \rho - \lambda + \lambda	\\
					&= \rho				\\
					\therefore \quad &1 + \lambda k - \rho k - |\Omega| \lambda + \lambda < \rho		\\
				\end{align*}
				\end{minipage}

				Thus all label classes receives between $\lambda$ and $\rho$ marginal probability.
				
				We finally show that the cost constraint is satisfied by $D$. We will show this via a transformation argument. Given that $D'$ is an improvising distribution, it must follow the randomness over words and randomness over labels constraints. In addition, as the distribution within a cost class does not effect the cost, we will only compare $D$ and $D'$ as distributions over the cost classes. The greedy LQCI construction assigns as much conditional probability to the cost classes with the lowest cost as allowed by our randomness over words requirement. Any more probability assigned to the lowest cost classes would violate these requirements, and any less would result in an equal or higher expected cost. Therefore, for each label $\ell_i$, we could change the conditional probability distribution $D_i'$ to the corresponding one from the greedy LQCI construction and have a lower or equal cost. We can perform a similar transformation for the marginal label distribution. The Greedy LQCI Construction assigns as much marginal probability to the label classes with lower expected cost as is allowed by the Randomness over Labels constraint. Again, any more probability to any of the lower cost label classes would violate these requirements, and any less would result in an equal or higher cost. Therefore, we can replace the marginal distribution in $D'$ with that of $D$ and have a lower or equal cost. Combining these two ideas, we can see that for any improvising distribution $D'$, the $D$ returned by the greedy LQCI construction has lower or equal cost. As $D'$ has low enough cost to be an improvising distribution and satisfy the cost constraint, $D$ must also satisfy the cost constraint. \qed

			\end{enumerate}
\end{proof}

Since the cost-minimality of the greedy LQCI construction that we have just shown will be useful later, we state it as a lemma:

\begin{lemma}
    \label{lemma:mincost}
    Let $\mathcal{C}$ be an LQCI instance. Then among distributions which satisfy the hard constraint, randomness over words constraint, and randomness over labels constraint (if they exist), the distribution returned by the greedy LQCI construction has minimal expected cost.
\end{lemma}

\textbf{Theorem \ref{thm:exact_scheme} Full Proof}

\begin{proof}
We recall the exact operations list in Def. \ref{def:sufficient_ops}:

    \begin{enumerate}
        \item Compute the finite list of possible costs $\Theta$.
        \item For each $i \in \{1, \dots, |\Omega|\}$ and each $k \in \Theta$ compute $|I_{i,k}|$.
        \item For each $i \in \{1, \dots, |\Omega|\}$ and each $k \in \Theta$, sample uniformly from the cost class $I_{i,k}$.
    \end{enumerate}

The idea of the proof is to use the above operations to efficiently construct the distribution returned by the greedy LQCI construction over the given parameters (or determine the instance is infeasible), find its expected cost, and (if feasible) sample from it.
Suppose we are given an LQCI instance from a class where the operaions above can be performed in polynomial time.
Because operation (\ref{op:compute_costs}) runs in polynomial time, the number of costs $|\Theta|$ must be at most polynomial in the size of the instance (and of course the number of labels as well since they are part of the instance).
We can compute the size of each label class $I_i$, of which there are polynomially many, using operations (\ref{op:compute_costs}) and (\ref{op:count_cost_class}) to measure each $|I_{i,k}|$ and applying $|I_i| = \sum_{k \in \Theta} |I_{i,k}|$.
Along with the parameters of the instance, this allows us to determine whether the first two conditions of Theorem \ref{thm:lqci_greedy} hold.
If not, then our algorithm returns that the instance is infeasible.
Otherwise, for each label $i$, we follow the greedy cost construction outlined in Section \ref{sec:theory} to assign probabilities to each cost class of the label, and compute the expected cost of the label class under those probabilities.
This computation requires evaluating a sum with a number of terms bounded by the number of cost classes (which must be polynomial), so it is polynomial.
Then we follow the greedy label construction to assign a marginal probability to each label class, a process which is also polynomial given the polynomial number of label classes.
Having now obtained a complete description of the greedy LQCI distribution, we compute its expected cost and see if it satisfies the cost requirement.
If it does not, then by Theorem \ref{thm:lqci_greedy} the instance is infeasible.
If it does, then the greedy distribution is an improvising distribution.
An improviser sampling from this distribution can then work as follows: we randomly pick a cost class $I_{i,k}$ with the probability assigned by the greedy distribution ($\hat{D}(I_i)D_i(I_{i,k})$), then apply operation (\ref{op:sample_cost_class}) to uniformly sample from it.
This improviser runs in time polynomial in the size of the LQCI instance, and constructing it took polynomial time, so this procedure is a polynomial-time improvisation scheme. \qed
\end{proof} 

\textbf{Theorem \ref{thm:output_dfa} Full Proof}

\begin{proof}
Using classical algorithms for DFAs, we can uniformly count and sample from the accepting words of a DFA of a particular length in time polynomial in the length and the size of the DFA's presentation~\cite{hickey:uniform_random_gen_cfg}, and we can create the product DFA from two DFAs in time multiplicative in the size of each component DFA \cite{DBLP:books/daglib/0016921}.
The number of different costs is bounded by the number of accepting states of the cost DFA, and so there are a linear number of elements in $\Theta$ which we can obtain in linear time to implement operation (\ref{op:compute_costs}).
For each $i \in \{1, \dots, |\Omega|\}$ (which we are given by the instance) and $k \in \{1, \dots, |\Theta|\}$ (just computed) we can adjust the cost and label DFAs to accept only words of label $i$ and cost $k$.
Then we can use the product construction to create a DFA which is the intersection of the hard constraint DFA and this DFA.
Then for each of the lengths between $m$ and $n$ (linearly-many since they are presented in unary), we can count the number of accepting words of that length for this DFA in polynomial time to implement operation (\ref{op:count_cost_class}).
For operation (\ref{op:sample_cost_class}), we can uniformly sample from this DFA in polynomial time.
So we can implement all operations required by Theorem~\ref{thm:exact_scheme} in polynomial time, and therefore this class of LQCI instances has a polynomial-time improvisation scheme. \qed
\end{proof}

\textbf{Theorem \ref{thm:acc_dfa} Full Proof}

\begin{proof}
    First we treat the case where the values of the costs for each state are encoded in unary, not binary. 
    For operation (\ref{op:compute_costs}), we apply a dynamic programming approach. Let $c(q,s)$ refer to the multiset of the costs of words of length $s$ that end in state $q$, and let $k(q)$ be the cost associated with entering state $q$. Initialize $c(q,0) = \{k(q)\}$ for the start state and $c(q,0) = \emptyset$ for all other states.
    Then build the table for $c(q,s)$ for $s > 0$ according to the recursion $c(q,s) = k(q) + \cup_{p\in Parent(q)} c(p,s-1)$, where addition is performed elementwise on the multiset. In other words, you take the union of the multisets of the parents of a state, and then add the cost of entering state $q$ to every element of the new multiset.
    The cost of the highest-cost word of length between $m$ and $n$ is bounded by the cost $M$ of the highest-cost state times $n+1$.
    Since the costs are integers, the number of distinct costs is bounded by $M(n+1)$, and as $M$ and $n$ are encoded in unary, this means that all the sets $c(q,s)$ have size at most polynomial in the size of the instance.
    Therefore computing these sets for all states in the DFA and all lengths up to $n$ takes polynomial time.

    To create the DFA for the cost class $I_{i,k}$, we take the product of the hard constraint DFA, the adjusted label DFA which only accepts label $i$, as well as a DFA which accepts only words of cost $k$.
    This last DFA can be built by taking the cost DFA and creating $k+1$ copies of each state to track the current accumulated cost up to a limit of $k+1$.
    Since $k$ is encoded in unary, the product DFA, whose size is the product of its components, is polynomial in the size of the instance.
    Counting and uniform sampling of words of the appropriate length from this DFA is done in polynomial time as in Theorem \ref{thm:output_dfa} to implement operations (\ref{op:count_cost_class}) and (\ref{op:sample_cost_class}).
    So we can implement all operations required by Theorem~\ref{thm:exact_scheme} in polynomial time, and therefore this class of LQCI instances (with costs encoded in unary) has a polynomial-time improvisation scheme.
    
    This immediately leads to the theorem, as moving from the binary encoding of our hypothesis to the unary coding above results in a pseudopolynomial scheme for the binary encoding. \qed

\end{proof}

The following two lemmas will be used to prove Theorem \ref{thm:approx_lqci}.

\begin{lemma}{(Approximate Greedy Cost Algorithm)}
\label{lemma:agca}
Given a Boolean-encoded LQCI instance $\CII$ and label class $i$ with associated $\phi_i(x,y,z)$ and $r$ as in Section \ref{sec:approx_lqci_alg}, let $\alpha$ and $\beta$ be the randomness parameters for the label.
Suppose approximate model counts are performed with confidence $1-\delta$ and tolerance $\tau$.
If Algorithm \ref{alg:approx_greedy_cost} returns on line \ref{line:agc_feasible}, let $D$ be the distribution obtained by picking a bucket according to the returned probabilities and then almost-uniformly sampling from that bucket with tolerance $\epsilon$ (in this case we say ``$D$ exists'').
Then with probability at least $(1-\delta)^b$ we will have the following guarantees:
\begin{enumerate}
    \item If $D$ exists, the expected cost of an item sampled from it satisfies $Lo \leq E[\CF(w) \mid w\leftarrow D] \leq rLo$.
    \item If $D$ exists, it assigns no word probability greater than $(1+\epsilon) (1+\tau)^2  \beta$ or less than $\alpha/(1+\epsilon)(1+\tau)^2$.
    \item If Algorithm 1 returns False, then $\CII$ is not feasible.
    \item If any distribution $D'$ over label class $i$ satisfies the randomness parameters $\alpha,\beta$ (i.e. no word is sampled with probability less than $\alpha$ or more than $\beta$) then $Lo \leq E[\CF(w) \mid w\leftarrow D']$.
\end{enumerate}
\end{lemma}
\begin{proof}
Since the approximate model counts are performed with confidence at least $1-\delta$ and tolerance $\tau$ and there are $b$ model counts, then with probability at least $(1-\delta)^{b}$ we have that for all $j$ that if $c_j$ is the approximate count for the $j$'th bucket and $c_j'$ the true count, then $c_j' /(1+\tau) \leq c_j \leq (1+\tau)c_j'$. We will assume this holds for the rest of the argument.\\

First observe that after being assigned in line (3), the $p_j$'s can only increase or stay the same in (6): as $\alpha \le \beta$ is required by the LQCI instance, $\alpha c_j/(1+\tau) \le (1+\tau)\beta c_j$. Then in (6), each $p_j$ is either increased to $(1+\tau)\beta c_j$, or is assigned $1-\sum_{j \neq k} p_j$, whichever is less. Moreover, only one $p_j$ is assigned $1-\sum_{j \neq k} p_j$ since afterwards the sum of the $p_j$'s is 1 and the loop breaks. Suppose now that at index $s \leq b$, $p_s$ is assigned $1-\sum_{j \neq s} p_j$. Then since after assignment $(\sum_{j \neq s} p_j) + \alpha c_s/(1+\tau) \leq 1$ but $(\sum_{j \neq s} p_j) + p_s = 1$, we have $p_s \geq \alpha c_s (1+\tau)$, as desired.\\

Next we show that if the algorithm returns values for the $p_j$'s, they form a distribution (which we call $D$), and that distribution's expected cost satisfies $Lo \leq E[\CF(w) \mid w\leftarrow D] \leq rLo$.

Since the algorithm did not terminate at (4), we know that $\sum_{j=1}^{b} p_j \leq 1$ at that point. Additionally, from (6) we know that when each $p_k$ re-appears, it can raise the total sum of the $p_j$'s to at most 1 (because each is set to a value $\leq 1 - \sum_{j \neq k} p_k$). Finally from (8) we know their sum cannot be less than 1, so their sum must be exactly 1 to reach (9) - where $Lo$ is assigned, and so the $p_j$'s form a probability distribution.

Then if the algorithm returns a distribution, it has expected cost $\sum_{j=1}^{b} p_j r_j'$, where $r_j'$ is the expected cost of almost-uniformly sampling the elements in bucket $j$ (wherefore we must have $r^{j-1} \leq r_j' < r^j$). Then we have $Lo = \sum_{j=1}^{b} p_j r^{j-1} \leq \sum_{j=1}^{b} p_j r_j' = E[\CF(w) \mid w\leftarrow D] \leq \sum_{j=1}^{b} p_j r^{j} = r Lo$.\\

Next, we show that if the algorithm returned a distribution, no element will be sampled with probability greater than $(1+\tau)^2(1+\epsilon)\beta$, and no element will be sampled with probability less than $\alpha/((1+\tau)^2(1+\epsilon))$.

As shown above, for all $j$, we have that $\alpha c_j/(1+\tau) \leq p_j \leq (1+\tau)\beta c_j$.
If we sampled exactly uniformly from a bucket with probability $p_j$, the chance of sampling any particular element would be $p_j/c_j'$, where $c_j'$ is the true count of bucket $j$.
Then by the previous, we have $\alpha c_j/((1+\tau)c_j') \leq p_j/c_j' \leq (1+\tau)\beta c_j/c_j'$. By our model counting argument in the first paragraph, we then have
\[ \frac{\alpha}{(1+\tau)^2} = \frac{\alpha c_j}{(1+\tau)^2 c_j} \leq \frac{\alpha c_j}{(1+\tau)c_j'} \leq \frac{p_j}{c_j'} \leq \frac{(1+\tau)\beta c_j}{c_j'} \leq \frac{(1+\tau)\beta c_j}{c_j/(1+\tau)} = (1+\tau)^2 \beta, \]
showing that probability of sampling any particular element lies between $\alpha/(1+\tau)^2$ and $(1+\tau)^2 \beta$, with perfectly uniform sampling.
Since in fact we perform almost-uniform sampling with tolerance $\epsilon$, this adds a multiplicative error of $(1+\epsilon)$ to each term.\\

Next we show that if the algorithm returned False, then LQCI instance is not feasible.
The algorithm can return False on lines (4) or (8).
Suppose it returns False on line (4).
Then we have $\sum_{j=1}^{b} \frac{\alpha}{1+\tau} c_j > 1$.
An improvising distribution would have to assign at least probability $\alpha$ to every element, and so would have to assign at least $\sum_{j=1}^{b} \alpha c_j'$ probability total.
However, by our model counting discussion earlier, we have $\sum_{j=1}^{b} \alpha c_j' \geq \sum_{j=1}^{b} \frac{\alpha}{1+\tau} c_j > 1$, requiring a probability greater than 1, which is impossible; so no improvising distribution exists and the instance is infeasible.
Similarly, suppose the algorithm returns False on line (8).
Then $\sum_{j=1}^{b} p_j <1$, which means that on line (6), each $p_k$ was assigned probability $(1+\tau)\beta c_k$ (since if any $p_k$ were assigned $1-\sum_{i \neq k} p_k$, the total probability would be 1 thereafter).
An improvising distribution would assign probability at most $\beta$ to each element, and so would be able to assign a maximum probability of $\sum_{j=1}^{b} \beta c_j'$.
Then once again by our model counting inequality we have $\sum_{j=1}^{b} \beta c_j' \leq \sum_{j=1}^{b} (1+\tau) \beta j_i < 1$, and so the instance is once again infeasible.\\

Finally we show that any distribution satisfying the randomness parameters $\alpha, \beta$ has expected cost at least $Lo$.
As discussed earlier, since $Lo$ has been assigned, the $p_j$'s form a probability distribution.
Moreover, each $p_j$ is at least $\alpha c_j/(1+\tau)$ and at most $(1+\tau)\beta c_j$. Let $G$ be the distribution returned by the greedy cost construction.
By cost-minimality of $G$ (Lemma \ref{lemma:mincost}), it suffices to show $Lo \leq E[\CF(w) \mid w\leftarrow G]$.
Let $p_j'$ be the probability with which the greedy distribution samples an element from bucket $j$, and let the expected cost of an item sampled from $G$ conditioned on sampling from bucket $j$ be $r_j'$.
We show that $Lo \leq \sum_{j=1}^{b} p_j' r_j'$.

Our algorithm starts by assigning $\alpha c_j /(1+\tau) $ probability to each bucket, and then increases the weight of each bucket starting with the cheapest to up to $(1+\tau)\beta c_j$ until a probability of 1 is reached. Let $m$ denote the number of times a bucket is assigned $(1+\tau) \beta c_j$ probability (so $0 \leq m \leq b)$. If $m=b$, then every bucket was assigned probability $(1+\tau)\beta c_j$. By the accuracy of the model counts, $p_j = (1+\tau)\beta c_j \geq \beta c_j' \geq p_j'$ for all $j$, and so $1 = \sum_{j=1}^{b} p_j \geq \sum_{j=1}^{b} p_j' = 1$, and thus $p_j = p_j'$ for all $j$. Then as $r^j \leq r_j' < r^{j+1}$, we have $Lo = \sum_{j=1}^{b} p_j r^j \leq \sum_{j=1}^{b} p_j' r_j'$, as desired.

Now suppose instead $m < b$. If $m>0$, then there will be a nonempty initial segment of buckets assigned $(1+\tau)\beta c_j$ probability, and thus as above an initial segment where for $j=1$ to $m$, $p_j \geq p_j'$.
If instead $m=0$, then bucket number 1 was assigned probability $1 - \sum_{j \neq 1} p_j$, and so the rest of the buckets are left with probability $\alpha c_i/(1+\tau)$.
By model counting, for all these buckets the greedy cost construction assigns at least $\alpha c_j' \geq \alpha c_j/(1+\tau) = p_j$, and so $\sum_{j=2}^{b} p_j' \geq \sum_{j=2}^{b} p_j$, which implies $p_1 \geq p_1'$ (as both distributions must sum to 1).
Therefore, in both cases, there is a length-$(1+m)$ segment $1, \dots , m$ for which $p_j \geq p_j'$. Then as $\sum_{j=1}^{b} p_j = \sum_{j=1}^{b} p_j' = 1$, we also have that $\sum_{j=1}^{m} p_j - p_j' = \sum_{j=m+1}^{b} p_j' - p_j \geq 0$. Then we have:
\begin{align*}
\sum_{j=1}^{b} p_j' r_j' - \sum_{j=1}^{b} p_j r^{j-1} &\geq \sum_{j=1}^{b} p_j' r^{j-1} - \sum_{j=1}^{b} p_j r^{j-1} \\
&= \left(\sum_{j=1}^{m} p_j' r^{j-1} + \sum_{j=m+1}^{b} p_j' r^{j-1}\right) - \left(\sum_{j=1}^{m} p_j r^{j-1} + \sum_{j=m+1}^{b} p_j r^{j-1}\right) \\
&= -\sum_{j=1}^{m} (p_j - p_j') r^{j-1} + \sum_{j=m+1}^{b} (p_j' - p_j) r^{j-1} \\
&\geq -\left(\sum_{j=1}^{m} p_j - p_j'\right)r^{m-1} + \left(\sum_{j=m+1}^{b} p_j' - p_j\right)r^m \\
&= -\left(\sum_{j=1}^{m} p_j - p_j'\right)r^{m-1} + \left(\sum_{j=1}^{m} p_j - p_j'\right)r^m \\
&= \left(\sum_{j=1}^{m} p_j - p_j'\right)(r^m - r^{m-1}) \\
&\geq 0
\end{align*}
So $Lo = \sum_{j=1}^{b} p_j r^{j-1}$ is less or than equal to the expected cost of the distribution $G$ returned by the greedy cost construction. This concludes the proof. \qed
\end{proof}

\begin{lemma}
\label{lemma:label_costerror}
Let $\CII$ be a Boolean LQCI instance. Let $D_i$ and $Lo_i$ be the distribution and $Lo$ value returned by running Algorithm 1 (with confidence $1-\delta$, tolerance $\tau$, and bucket width $r$) on label $i$. Suppose further that you run the greedy label construction using $Lo_i$ as label $i$'s expected cost, and $p_i$ is the probability assigned by the construction to label $i$. This results in a distribution $\widetilde{D}$ over words. Let $Low = \sum_{i=1}^{|\Omega|} p_i Lo_i$. Then the expected cost of $\widetilde{D}$ satisfies $Low \leq E[\CF(w)|w \leftarrow \widetilde{D} ] \leq r \cdot Low$. Moreover, for any improvising distribution $D'$ for the instance, we have that $Low \leq E[\CF(w) | w\leftarrow D']$.
\end{lemma}

\begin{proof} 
By Lemma \ref{lemma:agca}, we know that the distribution $D_i$ returned by running Algorithm 1 on label $i$ satisfies $Lo_i \leq E[\CF(w) | w \leftarrow D_i] \leq r Lo_i$ for all $i$. The expected cost of $\widetilde{D}$ is given by $E[\CF(w)| w\leftarrow \widetilde{D}] = \sum_{i=1}^{|\Omega|} p_i E[\CF(w) | w \leftarrow D_i]$, and so we have $Low = \sum_{i=1}^{|\Omega|} p_i Lo_i \leq \sum_{i=1}^{|\Omega|} p_i E[\CF(w) | w \leftarrow D_i] \leq \sum_{i=1}^{|\Omega|} p_i r Lo_i = r \cdot Low$, giving the first claim.\\

To show the second claim, it suffices to show that the distribution $G$ returned by the greedy LQCI construction from Section \ref{sec:theory} satisfies $Low \leq E[\CF(w) \mid w \leftarrow G]$, as $G$ has minimal cost among improvising distributions (Lemma \ref{lemma:mincost}). For each label we have an associated $Lo_i$ value (from Algorithm 1), and also the expected cost for sampling from that label according to the greedy construction, call this value $k_i$. Moreover, from Lemma \ref{lemma:agca} (4) we have that  $Lo_i \leq k_i$. Let $p_i'$ be the probability assigned to label $i$ with expected cost $k_i$ by the greedy construction. Recall that for a set of fixed label costs, the greedy label construction is optimal for minimizing expected costs, so in particular, $\sum_{i=1}^{|\Omega|} p_i' Lo_i \geq \sum_{i=1}^{|\Omega|} p_i Lo_i$. Using this we can show that $Low$ is a lower bound.
We have:
      \[ E[\CF(w)|w\leftarrow G] =  \sum_{i=1}^{|\Omega|} p_i' k_i \geq \sum_{i=1}^{|\Omega|} p_i' Lo_i \geq \sum_{i=1}^{|\Omega|} p_i Lo_i = Low \]
\qed
\end{proof}

\textbf{Theorem \ref{thm:approx_lqci} Proof of Correctness}
\begin{proof}
We take as input a confidence $1- \delta \in (0,1]$, randomness tolerance $\gamma > 0$, and cost tolerance $\zeta > 0$.
Fix the bucket ratio $r = 1 + \zeta$, and set the number of buckets $b = \ceil{\log_r (2^{|y|}))}$ where $|y|$ is the size in bits of the cost encoding.
Next form a Boolean formula $\phi_i (x,y,z)$ for each label class (of which there are polynomially many) as described in Section \ref{sec:approx_lqci_alg}, and run Algorithm 1 on each such formula with model count confidence $1 - d$ such that $(1-d)^{|\Omega|b} \geq 1 - \delta$, and model count tolerance $\tau$ such that $(1+\tau)^3 \leq 1 + \gamma$, i.e. such that we have confidence at least $1-\delta$ that all our model counts are accurate within a factor of $1+\tau$.
We assume this accuracy for the rest of the argument.
Now by Lemma \ref{lemma:agca}, running almost-uniform sampling with tolerance $\epsilon = \tau$ for each label $i$ we obtain a distribution $D_i$ over words in the label as well as a lower bound of that distribution's expected cost $Lo_i$. Then using the greedy label construction from Section \ref{sec:theory} over these $Lo_i$ values, we obtain a distribution over labels, and using this along with $D_i$ a distribution over words $\widetilde{D}$.
If any iteration of Algorithm 1 fails, or if the greedy label construction fails, or if the associated $Low$ value of $\widetilde{D}$ from Lemma \ref{lemma:label_costerror} exceeds the cost parameter $\cost$, we return failure.\\

(1) Since the words sampled are obtained by sampling from solutions to $\{ \exists z. \phi_i (x,y,z) \}_{i=1}^{|\Omega|}$, and each such solution satisfies the hard constraint $\HC$, if the procedure above returns a distribution $\widetilde{D}$, we trivially have $\Pr [ \HC(w)| w \leftarrow \widetilde{D}] = 1$.

(2) By Lemma \ref{lemma:label_costerror}, $E[\CF(w) | w \leftarrow \widetilde{D}] \leq r \cdot Low$, and if above procedure returns a distribution we also have $Low \leq c$, so we have $E[\CF(w) | w \leftarrow \widetilde{D}] \leq rc = (1 + \zeta) c$.

(3) Since the greedy label construction succeeded, we have $\forall i \in \{1, \dots, |\Omega| \}, \ \lambda \leq \Pr{[w \in I_i \mid w \leftarrow \widetilde{D}]} \leq \rho $.

(4) By Lemma \ref{lemma:agca} (3), for each distribution $D_i$ over a label $i$, we have that $\hat{\alpha}_i/((1+\epsilon)(1+\tau)^2) \leq \Pr[y =w | w\leftarrow D_i] \leq (1+\epsilon)(1+\tau)^2\hat{\beta}_i$. Since $(1+\epsilon)(1+\tau)^2 = (1+\tau)^3 \leq 1+\gamma$, this means that we have $\forall i \in \{1, \dots, |\Omega| \}, \ \forall y \in I_i$,
\[ \hat{\alpha}_i/(1+\gamma) \leq \Pr[y = w \mid w \in I_i, w \leftarrow \widetilde{D}] \leq (1+\gamma)\hat{\beta}_i \]

Moreover, if we returned failure, we argue that $\CII$ is infeasible (as above, with confidence at least $1-\delta$).
The procedure can return failure in three situations.
First, if Algorithm 1 returns failure, whereby using Lemma \ref{lemma:agca} (3) we have that the instance is infeasible.
Second, if the greedy label construction fails, which can only happen if condition (1) in Theorem \ref{thm:lqci_greedy} does not hold (and by the same theorem the instance is therefore infeasible).
Third, if $Low > \cost$, in which case by Lemma \ref{lemma:label_costerror} assuming there exists an improvising distribution $D'$ yields $\cost < Low \leq E[\CF(w)|w \leftarrow D']$, which is a contradiction as $D'$ violates the cost condition. \qed
\end{proof}

\textbf{Theorem \ref{thm:approx_lqci} Runtime}
\begin{proof}
Fix our desired confidence $1-\delta$, randomness tolerance $\gamma$, cost tolerance $\zeta$, and our formula size $|F|$.
As above, we need to choose the tolerance parameters $\epsilon,\tau$ such that $(1+\epsilon)(1+\tau)^2 \leq (1+\gamma)$, and for simplicity we take $\epsilon = \tau$.
Then it will suffice to use $\tau = \gamma/4$ for small enough $\gamma$, for example, so that $1/\tau = O(1/\gamma)$.
Similarly, we need to choose the confidence $1-d$ so that $(1-d)^{|\Omega| b} \ge 1 - \delta$ (with $b$ the number of buckets), which is possible with a $d$ of size satisfying $1/d = O(|\Omega|b/\delta)$.

The bucket ratio is $r = 1 + \zeta$ and the number of buckets for each label is $b = \ceil{\log_r(2^y)} = O(|F|/\log(r)) = O(|F|/\log(1+\zeta)) = O(|F|/\zeta)$, so there are $|\Omega|b = \textit{poly}(|\CII|, 1/\zeta)$ buckets in total.
Using \texttt{ApproxMC} \cite{approxmc_13} to count each bucket with confidence $1-d$ and tolerance $\tau$ takes $\textit{poly}(|F|, 1/\tau, \log(1/d))$ time.
Plugging in the asymptotic bounds for $1/\tau$ and $1/d$, the total time for counting all buckets is $\textit{poly}(|\CII|, 1/\zeta, 1/\gamma, \log(1/\delta))$.

Finally, the returned improviser can sample words in the same time, using \texttt{UniGen}~\cite{unigen_14} to sample approximately in time polynomial in $|F|$ and $1/\epsilon = O(1/\gamma)$.
(Or for sufficiently-small $\epsilon$, instead using the BGP algorithm \cite{bgp_sampling} to sample \emph{exactly} in time polynomial in $|F| = O(|\CII|)$; see footnote \ref{note:unigen_proviso}). \qed
\end{proof}

\textbf{Theorem \ref{thm:melqci_algorithm} Full Proof}
\begin{proof}
	This problem can be rephrased as an optimization problem.
	Note that the conditional distribution within a cost class does not affect the expected cost of a distribution and the maximum-entropy piecewise distribution is uniform over each piece.
	To this end we create an optimization variable $D(i,k)$ that represents the probability of selecting an element in the cost class $I_{i,k}$.
	With the assumption of a uniform conditional distribution within each cost class, which again we know to be entropy-maximizing, the total entropy of our resulting distribution $D$ will be:

	\begin{align*}		
    	H(D) &= -\sum_{w \in I} D(w) \log(D(w))	\\
    	&= -\sum_{i = 1}^{|\Theta|} \sum_{k=1}^{|\Omega|} \sum_{w \in I_{i,k}} D(w) \log(D(w))	 					\\
    	&= -\sum_{i = 1}^{|\Theta|} \sum_{k=1}^{|\Omega|} D(i,k) \log \left( \frac{D(i,k)}{|I_{i,k}|} \right) 	\\
    	&= -\sum_{i = 1}^{|\Theta|} \sum_{k=1}^{|\Omega|} D(i,k) \left( \log(D(i,k)) - \log(|I_{i,k}|) \right) 	\\
    	&= -\sum_{i = 1}^{|\Theta|} \sum_{k=1}^{|\Omega|} D(i,k) \log(D(i,k)) - D(i,k) \log(|I_{i,k}|) 	\\
    	&= \sum_{i = 1}^{|\Theta|} \sum_{k=1}^{|\Omega|}  -D(i,k) \log(D(i,k))+ D(i,k) \log(|I_{i,k}|)
	\end{align*}

	Adding optimization constraints to account for all the requirements of an improvising distribution, we get the following optimization problem:
    \begin{align*}
		&\text{Minimize} \
		-H(D) = - \sum_{i = 1}^{|\Theta|} \sum_{k=1}^{|\Omega|}  -D(i,k) \log(D(i,k))+ D(i,k) \log(|I_{i,k}|) 	\\
		&\text{subject to}	\\
		\tag{C1}  & \qquad \qquad \sum_{i = 1}^{|\Omega|} \sum_{k=1}^{|\Theta|} \mathcal{K}(I_{i,k}) D(i,k) - \cost  \leq 0			\\
		\tag{C2}  & \qquad \qquad \forall \ i \in \{ 1, \dots, |\Omega| \} , - \sum_{k = 1}^{|\Theta|} D(i, k) + \lambda \leq 0 	\\
		\tag{C3}  & \qquad \qquad \forall \ i \in \{ 1, \dots, |\Omega| \} , \ \sum_{k = 1}^{|\Theta|} D(i, k) - \rho \leq 0\\
		\tag{C4}  & \qquad \qquad \forall \ i \in \{ 1, \dots, |\Omega| \}, \forall \ k \in \{ 1, \dots, |\Theta| \}, -D(i,k) \leq 0					\\
		\tag{C5}  & \qquad \qquad \sum_{i = 1}^{|\Omega|} \sum_{k=1}^{|\Theta|} D(i,k) - 1 = 0\\
		\tag{C6}  & \qquad \qquad \forall \ i \in \{ 1, \dots, |\Omega| \}, \forall \ k \in \{ 1, \dots, |\Theta| \}, \ \text{if} \ |I_{i,k}| = 0, D(i,k) = 0 \\
	\end{align*}
	
	As this takes the form of an optimization problem minimizing a separable convex equation with linear constraints, we can solve this in polynomial time with only a logarithmic dependency on the maximum additive error $\tau$ allowed using an algorithm by Chubanov \cite{chubanov:convex_opt_poly_alg}.
	All we require is an initial feasible distribution, which can be calculated using the greedy LQCI construction (efficiently, according to Theorem~\ref{thm:exact_scheme}) for the equivalent LQCI problem.
	
	We must now show that these two problems are equivalent.
	The objective function being minimized is the negation of the entropy of the distribution defined by the $D(i,k)$ probabilities, as we saw above.
	So we only need show that a distribution assigning probability $D(i,k)$ to cost class $I_{i,k}$ (distributed uniformly over words in that class) is an improvising distribution if and only if it satisfies constraints \textbf{(C1-C6)}.
	
	\textbf{($\Rightarrow$) If a distribution is improvising for the MELQCI problem, it satisfies constraints \textbf{(C1-C6)}.}\\
	Let $D$ be an improvising distribution for a MELQCI problem $\CII$, and define $D(i,k) = \Pr [w \in I_{i,k} \mid w \leftarrow D]$.
	As $D$ is improvising, the expected cost of a word sampled from it is at most $c$, which implies \textbf{(C1)} is satisfied.
	$D$ must also have marginal label probabilities between $\lambda$ and $\rho$, which implies \textbf{(C2-C3)} are satisfied.
	Since $D$ is a probability distribution, all its probabilities are nonnegative and sum to 1, satisfying \textbf{(C4-C5)}.
	Finally, if $|I_{i,k}| = 0$ then no word is contained in $I_{i,k}$, so $D(i,k) = 0$ by definition and therefore \textbf{(C6)} is satisfied.
	
	\textbf{($\Leftarrow$) If a distribution satisfies constraints \textbf{(C1-C6)}, it is improvising for the MELQCI problem.}\\
	Given values $D(i,k)$ satisfying conditions \textbf{(C1-C6)}, the corresponding distribution on words $D$ is defined by letting $D(w) = D(i,k) / |I_{i,k}|$ if $w \in I_{i,k}$ (and $D(w) = 0$ if $w$ is not contained in any cost class).
	We need to show that $D$ is an improvising distribution.
	First, by conditions \textbf{(C4-C6)} the probabilities are nonnegative and sum to 1 (note how \textbf{(C6)} ensures that if $D(i,k) > 0$ then there is some $w \in I_{i,k}$), so $D$ is a probability distribution.
	Next, since we only give probability to elements of $I$, the hard constraint is satisfied.
	$D$ also satisfies the cost constraint due to condition \textbf{(C1)}, and the randomness over labels constraint due to conditions \textbf{(C2-C3)}.
	So $D$ is an improvising distribution for $\CII$. \qed
\end{proof}

\end{document}